\documentclass[runningheads]{llncs}

\usepackage{vmargin}

\setmarginsrb{3.75cm}{1.2cm}{3.75cm}{3.cm}{1.0cm}{1.6cm}{0.8cm}{1.4cm}

\usepackage{amssymb,amstext}
\usepackage{cite}
\usepackage{amsmath,dsfont}
\usepackage{fancyhdr}
\usepackage{boxedminipage}
\usepackage{colortbl}
\usepackage{hyperref}

\newcommand{\Eds}     [0] {\ensuremath{\mathcal{E}_\textsc{$r$DS}}\xspace}
\newcommand{\Cds}     [0] {\ensuremath{\mathcal{C}_\textsc{$r$DS}}\xspace}
\newcommand{\Lcds}    [0] {\ensuremath{L_{\Cds}}\xspace}

\newcommand{\Esc}     [0] {\ensuremath{\mathcal{E}_\textsc{$r$IS}}\xspace}
\newcommand{\Csc}     [0] {\ensuremath{\mathcal{C}_\textsc{$r$IS}}\xspace}
\newcommand{\Lcsc}    [0] {\ensuremath{L_{\Csc}}\xspace}

\newcommand{\EFD}     [0] {\ensuremath{\mathcal{E}_\textsc{$\cal{F}$D}}\xspace}
\newcommand{\CFD}     [0] {\ensuremath{\mathcal{C}_\textsc{$\cal{F}$D}}\xspace}
\newcommand{\LFD}    [0] {\ensuremath{L_{\CFD}}\xspace}
\newcommand{\EFDall} [0] {\ensuremath{\EFD=(\CFD,\LFD)}\xspace}

\newcommand{\rind}  [0] {\textsc{$r$-Independent Set}\xspace}

\newcommand{\eqec}  [0] {\ensuremath{\equiv_{\Eds,t}}\xspace}

\newcommand{\rdom}  [0] {\textsc{$r$-Dominating Set}\xspace}

\newcommand{\tw}{{\mathbf{tw}}}

\newcommand{\set}[1]{\ensuremath{\left\{#1\right\}}}
\newcommand{\poly}{\mathop{\rm poly}}

\usepackage{graphicx}

\newcommand{\rp}{{\bf rp}}

\renewcommand{\leq}{\leqslant}
\renewcommand{\geq}{\geqslant}
\renewcommand{\le}{\leqslant}

\newcommand{\mc}{\mathcal}

\newcommand{\ci}{~\cite}
\renewcommand{\sc}{\textsc}

\renewenvironment{proof}[1][]{\par \noindent {\bf Proof:#1}\ }{\hfill$\Box$\\}

\newenvironment{proofrDomSet}[1][]{\par \noindent {\bf Proof of Lemma~\ref{lem:rDomCharacterizes}:#1}\ }{\hfill$\Box$\\}

\newenvironment{proofDeletion}[1][]{\par \noindent {\bf Proof of Theorem~\ref{thm:KernelFDeletion}:#1}\ }{\hfill$\Box$\\}

\newenvironment{runningExample}{\par \medskip \noindent \emph{Running example:}}{\hfill$\diamond$\\}

\newcommand{\YES}{\textsc{Yes}}

\usepackage{xspace}

\newtheorem{fact}{Fact}

\newcommand{\C}{\mathcal{C}}

\author{Valentin Garnero\inst{1} \and Christophe Paul\inst{1} \and Ignasi Sau\inst{1} \and  Dimitrios M. Thilikos\inst{1}$^,$\inst{2}}

\title{Explicit linear kernels via dynamic programming\thanks{A short conference version of this article appeared in the \emph{Proc. of the 31st Symposium on Theoretical Aspects of Computer Science (STACS)},  volume 25 of LIPIcs, pages 312-324, Lyon, France, March 2014.
This work was supported by the ANR project AGAPE (ANR-09-BLAN-0159) and the Languedoc-Roussillon Project ``Chercheur d'avenir'' KERNEL. The fourth author was co-financed by the E.U. (European Social Fund - ESF) and Greek national funds through the Operational Program ``Education and Lifelong Learning'' of the National Strategic Reference Framework (NSRF) - Research Funding Program: ``Thales. Investing in knowledge society through the European Social Fund''. Emails of authors: {\sf Valentin.Garnero@lirmm.fr},
{\sf Christophe.Paul@lirmm.fr},
{\sf Ignasi.Sau@lirmm.fr},
{\sf sedthilk@thilikos.info}.}
}

\titlerunning{Explicit linear kernels via dynamic programming}

\authorrunning{Valentin Garnero, Christophe Paul,  Ignasi Sau, and Dimitrios M. Thilikos}

\institute{AlGCo project-team, CNRS, LIRMM, Montpellier, France. \and Department of Mathematics, \\National and Kapodistrian University of Athens, Athens, Greece.\\ \vspace{.15cm}}

\begin{document}

\maketitle
\setcounter{footnote}{0}

\begin{abstract}
Several algorithmic meta-theorems on kernelization have appeared in the last years, starting with the result of Bodlaender \emph{et al}. [FOCS 2009] on graphs of bounded genus, then generalized by Fomin \emph{et al}. [SODA 2010] to graphs excluding a fixed minor, and by Kim \emph{et al}. [ICALP 2013] to graphs excluding a fixed topological minor.  Typically, these results guarantee the existence of linear or polynomial kernels on sparse graph classes for problems satisfying some generic conditions but, mainly due to their generality, it is not clear
how to derive from them constructive kernels with explicit constants.

In this paper we make a step toward a fully constructive meta-kernelization theory on sparse graphs. Our approach is based on a more explicit protrusion replacement machinery that, instead of expressibility in CMSO logic,  uses dynamic programming, which allows us to find an explicit upper bound on the size of the derived kernels. We demonstrate the usefulness of our techniques by providing the first explicit linear kernels for \textsc{$r$-Dominating Set} and  \textsc{$r$-Scattered Set} on apex-minor-free graphs, and  for \textsc{Planar-$\mc{F}$-Deletion} on graphs excluding a fixed (topological) minor in the case where all the graphs in $\mc{F}$ are connected.

\vspace{0.25cm} \textbf{Keywords:} parameterized complexity, linear kernels, dynamic programming, protrusion replacement, graph minors. \end{abstract}

\section{Introduction}
\label{sec:intro}
\textbf{Motivation.} Parameterized complexity deals with problems whose instances $I$ come equipped  with an additional integer parameter $k$, and the objective is to obtain algorithms whose running time is of the form $f(k) \cdot \poly(|I|)$, where $f$ is some computable function (see~\cite{FlGr06,DF99} for an introduction to the field). We will be only concerned with problems defined on graphs. A fundamental notion in parameterized complexity is that of \emph{kernelization}, which asks for the existence of polynomial-time preprocessing algorithms that produce equivalent instances whose size depends exclusively (preferably polynomially or event linearly)  on $k$. Finding kernels of size polynomial or linear in $k$ (called \emph{linear kernels}) is one of the major goals of this area.

An influential work in this direction was the linear kernel of  Alber \emph{et al}.~\cite{AFN04} for \textsc{Dominating Set} on planar graphs, which was generalized by Guo and Niedermeier~\cite{GuNi07} to a family of problems on planar graphs. Several algorithmic meta-theorems on kernelization have appeared in the last years, starting with the result of Bodlaender \emph{et al}.~\cite{BFL+09} on graphs of bounded genus. After that, similar results have been obtained on larger sparse graph classes, such as graphs excluding a minor~\cite{FLST10} or a topological minor~\cite{KLP+12}.


Typically, the above results guarantee the {\sl existence} of linear or polynomial kernels on sparse graph classes for a number of problems satisfying some generic conditions but, mainly due to their generality, it is hard to derive from them {\sl constructive} kernels with {\sl explicit} constants. The main reason behind this non-constructibility is that the proofs rely on a property of problems called \emph{Finite Integer Index} (FII) that, roughly speaking, allows to replace large ``protrusions'' (i.e., large  subgraphs  with small boundary to the rest of the graph) with ``equivalent'' subgraphs of constant size. This substitution procedure is known as {\em protrusion replacer}, and while
its {\sl existence} has been proved, so far, there is no generic  way to {\sl construct} it. Using the technology developed in~\cite{BFL+09}, there are cases where protrusion replacements
can become constructive given the expressibility
of the problem in Counting Monadic Second Order (CMSO) logic. This approach is essentially based on extensions of
Courcelle's theorem~\cite{Cou90} that, even when they offer constructibility, it is hard to extract from them any {\sl explicit
constant} that upper-bounds the size of the derived kernel.\vspace{.25cm}


\noindent \textbf{Results and techniques.}
In this article we  tackle the above issues and make  a step toward a fully constructive meta-kernelization theory on sparse graphs with explicit constants. For this, we essentially substitute the algorithmic power of CMSO logic
with that of dynamic programming on graphs of bounded decomposability (i.e., bounded treewidth).
Our approach provides a  dynamic programming framework able to construct a protrusion
replacer for a wide variety of problems.

Loosely speaking, the framework that we present can be summarized as follows. First of all, we
propose a general definition of a problem encoding for the tables of dynamic programming when solving
parameterized problems on graphs of bounded treewidth. Under this setting, we
provide general conditions on whether such an encoding can yield a protrusion replacer.
While our framework can also be seen as a possible
formalization of dynamic programming, our purpose
is to use it for constructing protrusion replacement algorithms
and  linear kernels whose size is explicitly determined.

In order to obtain an explicit linear kernel for a problem $\Pi$, the main ingredient
is to prove that when solving $\Pi$ on graphs of bounded treewidth via dynamic programming, we can use tables such that the maximum difference between all the values that need to be stored is bounded by a function of the treewidth.
For this,
we prove in Theorem~\ref{thm:protrusionReplacement} that when the input graph excludes a fixed graph $H$ as a (topological) minor, this condition is sufficient  for constructing an explicit protrusion replacer algorithm, i.e., a polynomial-time algorithm that  replaces a large protrusion with an equivalent one whose size can be bounded by an {\sl  explicit} constant.
Such a protrusion replacer
can then be used, for instance, whenever it is possible
to compute a linear protrusion decomposition of the input graph (that is, an algorithm that partitions the graph into a part of size linear in $O(k)$ and a set of $O(k)$ protrusions). As there is a wealth of
results for constructing such decompositions~\cite{BFL+09,FLST10,KLP+12,FLMS12}, we can use them as a starting point
and, by applying dynamic programming, obtain an explicit linear kernel for $\Pi$.

We demonstrate the usefulness of this general strategy by providing the first explicit linear kernels for three distinct families of problems on sparse graph classes. On the one hand,  for each integer $r\geq 1$, we provide a linear kernel for \textsc{$r$-Dominating Set} and  \textsc{$r$-Scattered Set} on graphs excluding a fixed apex graph $H$ as a minor. Moreover, for each finite family $\mc{F}$ of connected graphs containing at least one planar graph, we provide a linear kernel for \textsc{Planar-$\mc{F}$-Deletion} on graphs excluding a fixed graph $H$ as a (topological) minor\footnote{In an earlier version of this paper, we also described a linear kernel for \textsc{Planar-$\mc{F}$-Packing} on graphs excluding a fixed graph $H$ as a minor. Nevertheless, as this problem is not directly vertex-certifiable (see Definition~\ref{def:vertexCertifiable}), for presenting it we should restate and extend many of the definitions and results given in Section~\ref{sec:framework} in order to deal with more general families of problems. Therefore, we decided not to include this family of problems in this article.}.
We chose these families of problems as they are all tuned by a secondary parameter that is
either the constant $r$ or the size of the graphs in the family ${\cal F}$. That way, we not only
capture a wealth of parameterized problems, but we also make explicit the contribution
of the secondary parameter in the size of the derived kernels.  (We would like to note that the constants involved in the kernels for \textsc{$r$-Dominating Set} and  \textsc{$r$-Scattered Set} (resp. \textsc{Planar-$\mc{F}$-Deletion}) depend on the function $f_c$ (resp. $f_m$) defined in Proposition~\ref{prop:tw-contraction} (resp. Proposition~\ref{prop:tw-minor}) in Section~\ref{sec:prelim}.)

\vspace{.25cm}

\noindent \textbf{Organization of the paper.} For the reader not familiar with the background used in previous work on this topic~\cite{BFL+09,FLST10,KLP+12}, some preliminaries can be found in Section~\ref{sec:prelim}, including graph minors, parameterized problems, (rooted) tree decompositions, boundaried graphs, the canonical equivalence relation $\equiv_{\Pi,t}$ for a problem $\Pi$ and an integer $t$, FII, protrusions, and protrusion decompositions. In Section~\ref{sec:framework} we introduce the basic definitions of our framework and present an explicit protrusion replacer.  The next three sections are devoted to showing how to apply our methodology to various families of problems, Namely, we focus on \textsc{$r$-Dominating Set} in Section~\ref{sec:rDomSet}, on  \textsc{$r$-Scattered Set}  in Section~\ref{sec:rScatSet}, and on \textsc{Planar-$\mc{F}$-Deletion} in Section~\ref{sec:PlanarFDeletion}. Finally, we conclude with some directions for further research in Section~\ref{sec:further}.


\section{Preliminaries}
\label{sec:prelim}

\textbf{Graphs and minors.} We use standard graph-theoretic notation (see~\cite{Die05} for any undefined terminology). Given a graph $G$, we let $V(G)$
denote its vertex set and $E(G)$ its edge set.
For~$X \subseteq V(G)$, we let~$G[X]$ denote the graph $(X,E_X)$, where
$E_X := \set{xy \mid x,y \in X \ \text{and} \ xy \in E(G)}$, and we define $G-X:=G[V(G)
\setminus X]$. The open (resp. closed) \emph{neighborhood} of a vertex $v$ is denoted by $N(v)$ (resp. $N[v]$), and more generally, for an integer $r \geq 1$, we denote by $N_r(v)$ the set of vertices that are at distance at most $r$ from $v$. The neighborhoods of a set of vertices $S$ are defined analogously. The \emph{distance} between a vertex $v$ and a set of vertices $S$ is defined as $d(v,S)=\min_{u \in S}d(v,u)$, where $d(v,u)$ denotes the usual distance. A graph $G=(E,V)$ is an \emph{apex graph} if there exists $v \in V$ such that $G-v$ is planar. Given an edge~$e = xy$ of a graph~$G$, we let~$G/e$ denote the graph obtained from~$G$ by \emph{contracting} the edge~$e$, which amounts to deleting the
endpoints of~$e$, introducing a new vertex~$v_{xy}$, and making it adjacent to
all vertices in $(N(x) \cup N(y)) \setminus \{x,y\} $. A \emph{minor} (resp. \emph{contraction})
of~$G$ is a graph obtained from a subgraph of~$G$ (resp. from $G$) by contracting zero or more
edges. A \emph{topological minor} of~$G$ is a
graph obtained from a subgraph of~$G$ by contracting zero or more edges, such
that each edge that is contracted  has at least one endpoint with degree at
most two.
A graph $G$ is {\em $H$-(topological-)minor-free} if $G$ does not contain $H$ as a (topological) minor.


\vspace{.4cm}

\noindent \textbf{Parameterized problems, kernels, and treewidth.} A \emph{parameterized graph problem}~$\Pi$ is a set
    of pairs $(G,k)$, where $G$ is a graph and $k \in \mathbb{Z}$,
    such that for any two instances $(G_1,k_1)$ and $(G_2,k_2)$ with $k_1,k_2 < 0$ it holds that $(G_1,k_1) \in \Pi$ if and only if $(G_2,k_2) \in \Pi$.
    If $\mathcal G$ is a graph class, we define $\Pi$ \emph{restricted to}
    $\mathcal G$ as $\Pi_{\mathcal G} = \set{(G,k) \mid (G,k) \in
    \Pi ~\textnormal{and}~ G \in \mathcal G}.$
A \emph{kernelization algorithm}, or just \emph{kernel}, for a parameterized graph problem
    $\Pi$
    is an algorithm that given an instance $(G,k)$ outputs,
    in time polynomial in $|G| + k$, an instance $(G',k')$ of $\Pi$
    such that $(G,k) \in \Pi$ if and only if $(G',k') \in \Pi$ and $|G'|, k' \leq g(k)$,
    where~$g$ is some computable function. The function $g$ is called
    the \emph{size} of the kernel. If $g(k) = k^{O(1)}$ or $g(k) = O(k)$,
    we say that $\Pi$ admits a \emph{polynomial kernel} and a \emph{linear kernel}, respectively.

    Given a graph~$G=(V,E)$, a \emph{tree decomposition of $G$} is an ordered pair
    $(T, \mc{X}=\{ B_x \mid x \in V(T) \})$, where~$T$ is a tree and such that the following hold:
    \begin{enumerate}
        \item[(i)] $\bigcup_{x \in V(T)} B_x = V(G)$;
        \item[(ii)] for every edge~$e = uv$ in~$G$, there exists~$x \in V(T)$ such that~$u,v \in B_x$; and
        \item[(iii)] for each vertex~$u \in V(G)$, the set of nodes~$\{x \in V(T) \mid u \in B_x\}$         induces a subtree.
    \end{enumerate}
    The vertices of the tree~$T$ are usually referred to as \emph{nodes} and the sets~$B_x$
    are called \emph{bags}. The \emph{width} of a tree decomposition is the size of a largest
    bag minus one. The \emph{treewidth} of~$G$, denoted~$\tw(G)$, is the smallest width of a tree decomposition of~$G$. A \emph{rooted tree decomposition} is a tree decomposition $(T, \mc{X}=\{ B_x \mid x \in V(T) \})$ in which a distinguished node $r \in V(T)$ has been selected as the \emph{root}. The bag $B_r$ is called the \emph{root-bag}. Note that the root defines a child/parent relation between every pair of adjacent nodes in $T$, and ancestors/descendants in the usual way. A node without children is called a \emph{leaf}.

     For the definition of {\em nice tree decompositions}, we refer
to~\cite{Klo94}. A set of vertices $X$ of a graph $G$ is called a \emph{treewidth-modulator} if $\tw(G-X)\leq t$, where $t$ is some fixed constant.

Given a bag $B$ of a rooted
tree decomposition with tree $T$, we denote by $T_B$ the subtree rooted
at the node corresponding to bag $B$, and by $G_B := G[\bigcup_{x\in T_B}B_x]$
the subgraph of $G$ induced by the vertices appearing in the bags corresponding
to the nodes of $T_B$. If a bag $B$ is associated with a node $x$ of $T$, we may interchangeably use $G_B$ or $G_x$.






\vspace{.4cm}

\noindent \textbf{Boundaried graphs and canonical equivalence relation.} The following two definitions are taken from~\cite{BFL+09}.

\begin{definition}[Boundaried graphs]\label{def:boundaried}
A \emph{boundaried graph} is a graph $G$ with
a set $B \subseteq V (G)$ of distinguished vertices and an injective labeling $\lambda: B \to \mathds{N}^+$ . The set $B$ is called the \emph{boundary} of $G$ and the vertices in $B$ are called \emph{boundary vertices}. Given a boundaried graph $G$, we denote
its boundary by $\partial(G)$, we denote its labeling by $\lambda_G$, and we define its label set by
$\Lambda(G) = \{\lambda_G(v) \mid v \in \partial(G) \}$.
We say that a boundaried graph is a \emph{$t$-boundaried graph} if $\Lambda(G) \subseteq \{1, \ldots ,t\}$.  \end{definition}


Note that a $0$-boundaried graph is just a graph with no boundary.

\begin{definition}[Gluing operation]\label{def:gluing}
Let $G_1$ and $G_2$ be two boundaried graphs. We
denote by $G_1 \oplus G_2$ the graph obtained by taking the disjoint union of $G_1$ and $G_2$ and identifying  vertices with the same label of the boundaries of $G_1$ and $G_2$. In $G_1 \oplus G_2$ there is an edge between two labeled vertices if there is an edge between them in $G_1$ or in $G_2$.
\end{definition}

In the above definition, after identifying vertices with the same label, we may consider the resulting graph as a boundaried graph or not, depending on whether we need the labels for further gluing operations.


Following~\cite{BFL+09}, we introduce a canonical equivalence relation on boundaried graphs.

\begin{definition}[Canonical equivalence on boundaried graphs]\label{def:canonicalEquivRelation}
Let $\Pi$ be
a parameterized graph problem
and let $t \in \mathds{N}^{+}$.  Given
two $t$-boundaried graphs $G_1$ and $G_2$, we say that $G_1 \equiv_{\Pi,t} G_2$ if $\Lambda(G_1) = \Lambda(G_2)$ and there
exists a transposition constant $\Delta_{\Pi,t}(G_1,G_2) \in \mathds{Z}$ such that for every $t$-boundaried graph $H \in \mc{G}$ and every $k \in \mathds{Z}$, it holds that $(G_1 \oplus H, k) \in \Pi$ if and only if $(G_2 \oplus H, k+\Delta_{\Pi,t}(G_1,G_2)) \in \Pi$.
\end{definition}

We define in Section~\ref{sec:framework} another equivalence relation on boundaried graphs that refines this canonical one (cf. Definitions~\ref{def:equivalenceRelation} and~\ref{def:equivalenceRelation2}), and that will allow us to perform a constructive protrusion replacement with explicit bounds.

The notion of \emph{Finite Integer Index} was originally defined by Bodlaender and van Antwerpen-de Fluiter\ci{BvF01,vanF97}. We would like to note that FII does not play any role in the framework that we present for constructing explicit kernels, but we present its definition for completeness, as we will sometimes refer to it throughout the article.

\begin{definition}[Finite Integer Index (FII)]\label{def:FII}
A parameterized graph problem $\Pi$
has \emph{Finite Integer Index} (\emph{FII} for short) if for every positive integer $t$, the equivalence relation $\equiv_{\Pi,t}$ has finite index.
\end{definition}

\vspace{.15cm}

\noindent \textbf{Protrusions and protrusion decompositions.}
Given a graph~$G=(V,E)$ and a set~$W \subseteq V$, we define~$\mathbf{bd}(W)$
as the vertices in~$W$ that have a neighbor in~$V \setminus W$.
A set~$W \subseteq V(G)$ is a \emph{$t$-protrusion}
     if $|\mathbf{bd}(W)| \leq t$ and $\tw(G[W]) \le t-1$.
We would like to note that a $t$-protrusion $W$ can be naturally seen as a $t$-boundaried graph by arbitrarily assigning labels to the vertices in $\mathbf{bd}(W)$. In this case, it clearly holds that $\partial(W)=\mathbf{bd}(W)$. Note also that if $G$ is a $t$-boundaried graph of treewidth at most $t-1$, we may assume that the boundary vertices are contained in any specified bag of a tree decomposition, by increasing the width of the given tree decomposition to at most $2t-1$.

An $(\alpha,t)${\em -protrusion decomposition} of a graph $G$ is a partition
    ${\cal P}=Y_{0}\uplus Y_{1}\uplus \cdots\uplus Y_{\ell}$ of $V(G)$ such
    that:
    \begin{enumerate}
    \item[(i)] for every $1\leqslant i\leqslant \ell$, $N(Y_{i})\subseteq Y_{0}$;
    \item[(ii)] $\max\{\ell, |Y_{0}|\}\leqslant \alpha$; and
    \item[(iii)] for every $1\leqslant i\leqslant \ell$, $Y_i\cup N_{Y_0}(Y_i)$  is a $t$-protrusion of $G$.
    \end{enumerate}
When $G$ is the
input of a parameterized graph problem with parameter $k$, we say that an
$(\alpha,t)$-protrusion decomposition of $G$ is \emph{linear} whenever $\alpha =O(k)$.

\vspace{.45cm}

\noindent \textbf{Large treewidth and grid minors.} In our applications in Sections~\ref{sec:rDomSet}, \ref{sec:rScatSet}, and \ref{sec:PlanarFDeletion} we will need the following results, which state a {\sl linear} relation between the treewidth and certain grid-like
graphs  that can be found as  minors or  contractions in a graph that excludes some fixed (apex) graph as a minor.
%


\begin{proposition}[Demaine  and Hajiaghayi~\cite{DemaineH08line}]\label{prop:tw-minor}
There is a function $f_m:\mathds{N}\rightarrow\mathds{N}$ such that for every $h$-vertex graph $H$
and every positive integer $r$, every $H$-minor-free graph with treewidth at least $f_{m}(h)\cdot r$, contains
an $(r\times r)$-grid as a minor.
\end{proposition}

Before we state the next proposition, we need to define a grid-like graph that is suitable for a contraction counterpart of Proposition~\ref{prop:tw-minor}.
Let ${\rm \Gamma}_{r}$  ($r\geq 2$) be the graph obtained from the  $(r\times r)$-grid by
triangulating internal faces of the $(r\times r)$-grid such that all internal vertices become  of degree $6$,
all non-corner external vertices are of degree 4,
and  one corner of degree 2 is joined by edges with all vertices
of the external face (the {\em corners} are the vertices that in the underlying grid have  degree 2).
The graph $\Gamma_6$ is shown in Fig.~\ref{fig-gamma-k}.

\begin{figure}[ht]
\vspace{-.45cm}
\begin{center}
\scalebox{.8}{\includegraphics{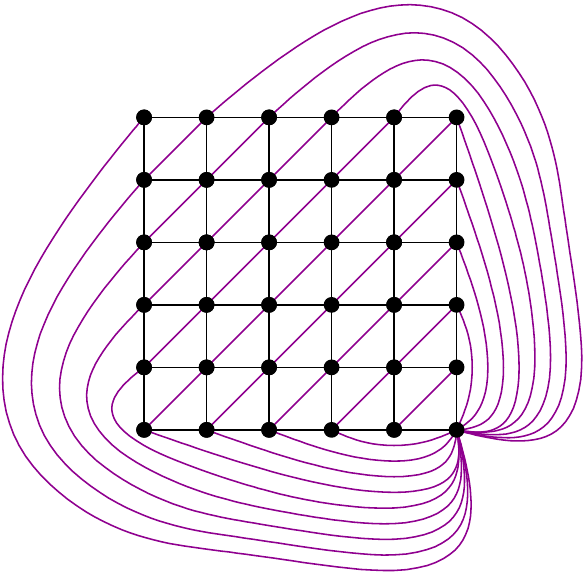}}
\end{center}
\caption{The graph $\Gamma_{6}$.}
\label{fig-gamma-k}
\end{figure}

\begin{proposition}[Fomin {\em et al.} \cite{FGT11}]\label{prop:tw-contraction}
There is a function $f_c:\mathds{N}\rightarrow\mathds{N}$ such that for every $h$-vertex apex graph $H$
and every positive integer $r$, every $H$-minor-free graph with treewidth at least $f_{c}(h)\cdot r$, contains
the graph  ${\rm \Gamma}_{r}$ as a contraction.
\end{proposition}

Propositions~\ref{prop:tw-minor} and~\ref{prop:tw-contraction} have been the main tools for developing  Bidimensionality
theory for kernelization~\cite{FLST10}. The best known estimation for~function $f_m$ has been given by Kawarabayashi and
Kobayashi in~\cite{KaKo12} and is $f_{m}(r)=2^{O(r^2\cdot \log r)}$. To our knowledge, no reasonable estimation for the function
$f_{c}$ is known up to now. The two functions $f_{m}$ and $f_{c}$ will appear in the upper bounds
on the size of the kernels presented in Sections~\ref{sec:rDomSet}, \ref{sec:rScatSet}, and \ref{sec:PlanarFDeletion}.
Any  improvement on these functions will directly translate to the sizes of our kernels.

%
%
%
%

\section{An explicit protrusion replacer}
\label{sec:framework}
In this section we present our strategy to construct an explicit protrusion replacer via dynamic programming. For a positive integer $t$, we define $\mathcal{B}_t$ as the class of all $t$-boundaried graphs, and we define $\mathcal{F}_t$ as the class of all $t$-boundaried graphs of treewidth at most $t-1$ that have a rooted tree decomposition with all boundary vertices contained in the root-bag\footnote{Note that the latter condition in the definition of $\mathcal{F}_t$ could be avoided by allowing the width of the tree decompositions of the graphs in $\mathcal{F}_t$ to be at most $2t-1$, such that all boundary vertices could be added to all bags of any tree decomposition.}. Note that it holds clearly that $\mathcal{F}_t \subseteq \mathcal{B}_t$. We will restrict ourselves to parameterized graph problems such that a solution can be certified by a subset of vertices.

\begin{definition}[Vertex-certifiable problem]\label{def:vertexCertifiable}
A parameterized graph problem $\Pi$ is called \emph{vertex-certifiable} if there exists a
language $L_{\Pi}$ (called {\em certifying language for $\Pi$}) defined on pairs $(G,S)$, where $G$ is a graph and $S \subseteq V(G)$, such that $(G,k)$ is a \YES-instance of $\Pi$ if and only if there exists a subset $S \subseteq V(G)$ with $|S| \leq k$ (or $|S| \geq k$, depending on the problem) such that $(G,S) \in L_{\Pi}$.
\end{definition}

Many graph problems are vertex-certifiable, like \sc{$r$-Dominating Set}, \sc{Feedback Vertex Set}, or \sc{Treewidth-$t$ Vertex Deletion}. This section is structured as follows. In Subsection~\ref{sec:encoders} we define the notion of encoder, the main object that will allow us to formalize in an abstract way the tables of dynamic programming. In Subsection~\ref{sec:equivalencerelations} we use encoders to define an equivalence relation on graphs in $\mc{F}_t$ that,
  under some natural technical conditions, will be a \emph{refinement} of the canonical equivalence relation defined by a problem $\Pi$ (see Definition~\ref{def:canonicalEquivRelation} in Section~\ref{sec:prelim}). This refined equivalence relation allows us to provide an explicit upper bound on the size of its representatives (Lemma~\ref{lem:finiteSize}), as well as a linear-time algorithm to find them (Lemma~\ref{lem:compute}). In Subsection~\ref{sec:explicitprotrusionreplacer} we use the previous ingredients to present an explicit protrusion replacement rule (Theorem~\ref{thm:protrusionReplacement}), which replaces a large enough protrusion with a bounded-size representative from its equivalence class, in such a way that the parameter does not increase.


\subsection{Encoders}
\label{sec:encoders}

The \textsc{Dominating Set} problem, as a vertex-certifiable problem, will be used hereafter as a running example to particularize our general framework and definitions. Let us start with a description of  dynamic programming tables for \textsc{Dominating Set} on graphs of bounded treewidth, which will illustrate the final purpose of the definitions stated below.

\begin{runningExample}\label{example:DominatingSet}
Let $B$ be a bag of a rooted tree decomposition $(T,\mathcal{X})$ of width $t-1$ of a graph $G \in \mc{F}_t$. The dynamic programming (DP) tables for \textsc{Dominating Set} can be defined as follows. The entries of the DP-table for $B$ are indexed by the set of tuples $R\in\{0,\uparrow 1,\downarrow 1\}^{|B|}$, so-called \emph{encodings}. As detailed below, the symbol 0 stands for vertices in the (partial) dominating set, the symbol $\downarrow 1$ for vertices that are already dominated, and $\downarrow 1$ for vertices with no constraints. More precisely, the coordinates of each $|B|$-tuple are in one-to-one correspondence with the vertices of $B$. For a vertex $v\in B$, we denote by $R(v)$ its corresponding coordinate in the encoding $R$. A subset $S \subseteq V(G_B)$ is a \emph{partial dominating set satisfying} $R$ if the following conditions are satisfied:
\begin{itemize}
\item[$\bullet$]  $\forall v\in V(G_B)\setminus B$, $d_{G_B}(v,S)\leqslant 1$; and
\item[$\bullet$]  $\forall v\in B$: $R(v)=0$ $\Rightarrow$ $v\in S$,  and  $R(v)=\downarrow 1$ $\Rightarrow$ $d_{G_B}(v,S)\leqslant 1$.
\end{itemize}
Observe that if $S$ is a partial dominating set satisfying $R$, then $\{v\in B\mid R(v)=0\}\subseteq S$, but $S$ may also contain vertices with $R(v)\neq 0$. Likewise, the vertices that are not (yet) dominated by $S$ are contained in the set $\{v\in B\mid R(v)=\uparrow 1\}$.
\end{runningExample} \vspace{-.10cm}

The following definition considers the tables of dynamic programming in an abstract way.



\begin{definition}[Encoder]\label{def:encoder}
An \emph{encoder} $\mc{E}$ is a pair $(\mc{C},L_{\mc{C}})$ where
 \begin{itemize}
 \item[(i)] $\mathcal{C}$ is a function that, for each (possibly empty) finite subset $I \subseteq \mathds{N}^+$, outputs a (possibly empty) finite set $\mathcal{C}(I)$ of strings over some alphabet. Each $R\in\mathcal{C}(I)$ is called a {\em $\mathcal{C}$-encoding} of $I$; and

 \item[(ii)] $L_\mathcal{C}$ is a computable language whose strings encode
  triples $(G,S,R)$, where $G$ is a boundaried graph, $S \subseteq V(G)$, and $R \in \mathcal{C}(\Lambda(G))$. If $(G,S,R)\in L_{\mathcal{C}}$, we say that $S$ \emph{satisfies} the $\mathcal{C}$-encoding $R$.
 \end{itemize}
\end{definition}


\noindent As it will become clear with the running example, the
set $I$ represents the labels from a bag, $\mathcal{C}(I)$ represents the possible configurations of the vertices in the bag, and $L_{\mathcal{C}}$ contains triples that correspond to solutions to these configurations.

\begin{runningExample}
Each rooted graph $G_B$ can be naturally viewed as a $|B|$-boundaried graph such that $B=\partial(G_B)$ with $I=\Lambda(G_B)$. Let $\mc{E}_{\textsc{DS}}=(\mathcal{C}_{\textsc{DS}},L_{\mathcal{C}_{\textsc{DS}}})$ be the encoder described above for \textsc{Dominating Set}. The tables of the dynamic programming algorithm to solve \textsc{Dominating Set} are obtained by assigning to every $\mathcal{C}_{\textsc{DS}}$-encoding (that is, DP-table entry) $R\in\mathcal{C}_{\textsc{DS}}(I)$,  the size of a minimum partial dominating set satisfying $R$, or $+\infty$ if such a set of vertices does not exist. This defines a function $f_{G}^{\mc{E}_{\textsc{DS}}}: \mathcal{C}_{\textsc{DS}}(I) \to \mathds{N} \cup \{+\infty\}$. Observe that if $B=\partial(G_B)=\emptyset$, then the value assigned to the encodings in $\mathcal{C}_{\textsc{DS}}(\emptyset)$ is indeed the size of a minimum dominating set of $G_B$.
\end{runningExample}

\vspace{-.05cm}

In the remainder of this subsection we will state several definitions for minimization problems, and we will restate them for maximization problems whenever some change is needed.  For a general minimization problem $\Pi$, we will only be interested in encoders that permit to solve $\Pi$  via dynamic programming. More formally, we define a $\Pi$-encoder and the values assigned to the encodings as follows.



\begin{definition}[$\Pi$-encoder and its associated function] \label{def:Piencoder}
 Let $\Pi$ be a vertex-certifiable minimization problem.  
 \begin{itemize} 
  \item[(i)] An encoder $\mc{E}=(\mc{C},L_{\mc{C}})$ is a \emph{$\Pi$-encoder} if $\mc{C}(\emptyset)$ consists of a single  $\mathcal{C}$-encoding, namely $R_{\emptyset}$, such that for every $0$-boundaried graph $G$ and every $S \subseteq V(G)$, $(G,S,R_{\emptyset}) \in L_{\mc{C}}$ if and only if $(G,S) \in L_{\Pi}$.
  \item[(ii)] Let $G$ be a $t$-boundaried graph with $\Lambda(G)=I$. We define the function $f_{G}^{ \mathcal{E}}: \mathcal{C}(I) \to \mathds{N} \cup \{+\infty\}$ as
\begin{equation}\label{eq:fEmin}
f_{G}^{ \mathcal{E}}(R) \ = \ \min \{k \ : \ \exists S \subseteq V(G), |S| \leq k, (G,S,R) \in L_{\mathcal{C}}\}.
\end{equation}

In Equation~(\ref{eq:fEmin}), if such a set $S$ does not exist, we set $f_{G}^{ \mathcal{E}}(R):= +\infty$. We define $\mathcal{C}_{G}^{*}(I):=\{ R \in \mathcal{C}(I) \mid f_{G}^{ \mathcal{E}}(R) \neq +\infty\}$.
\end{itemize}
\end{definition}

Condition~(i) in Definition~\ref{def:Piencoder} guarantees that, when the considered graph $G$ has no boundary, the language of the encoder is able to {\sl certify} a solution of  problem $\Pi$. In other words, we ask that the set $\{(G,S)\mid (G,S,R_{\emptyset})\in L_{\cal C})\}$ is a {\sl certifying language} for $\Pi$. Observe that for a $0$-boundaried graph $G$, the function $f_{G}^{ \mathcal{E}}(R_{\emptyset})$ outputs the minimum size of a set $S$ such that $(G,S)\in L_{\Pi}$.



For encoders $\mc{E}'=(\mc{C}',L_{\mc{C}'})$ that will be associated with problems where the objective is to find a set of vertices of size {\sl at least} some value, the corresponding function $f_{G}^{ \mathcal{E}'}: \mathcal{C}'(I) \to \mathds{N} \cup \{-\infty\}$ is defined as
\begin{equation}\label{eq:fEmax}
f_{G}^{ \mathcal{E}'}(R) \ = \ \max \{k \ : \ \exists S \subseteq V(G), |S| \geq k, (G,S,R) \in L_{\mathcal{C}'}\}.
\end{equation}
Similarly, in Equation~(\ref{eq:fEmax}), if such a set $S$ does not exist, we set $f_{G}^{ \mathcal{E}}(R):= -\infty$.
We define $\mathcal{C}_{G}^{*}(I):=\{ R \in \mathcal{C}(I) \mid f_{G}^{ \mathcal{E}}(R) \neq -\infty\}$.

The following definition provides a way to control the number of possible distinct values assigned to encodings. This property will play a similar role to FII or \emph{monotonicity} in previous work~\cite{BFL+09,KLP+12,FLST10}.

\begin{definition}[Confined encoding]\label{def:thin}
An encoder $\mathcal{E}$ is \emph{$g$-confined} if there exists a function $g : \mathds{N} \to \mathds{N}$ such that for any $t$-boundaried graph $G$ with $\Lambda(G) = I$ it holds that either $\mathcal{C}_{G}^{*}(I)= \emptyset $ or \vspace{-.2cm}
\begin{equation}\label{eq:thin}
\max_{R \in \mathcal{C}_{G}^{*}(I)}f_{G}^{ \mathcal{E}}(R)\  -\   \min_{R \in \mathcal{C}_{G}^{*}(I)}f_{G}^{ \mathcal{E}}(R)  \  \leq \ g(t).
\end{equation}
\end{definition}\vspace{-.2cm}

See Fig.~\ref{fig:illustrationEncoding} for a schematic illustration of a confined encoder. In this figure, each column of the table corresponds to a $\mathcal{C}$-encoder $R$, which is filled with the value $f_{G}^{ \mathcal{E}}(R)$.

\begin{figure}[h!]
\centering 
\includegraphics[width=0.8\textwidth]{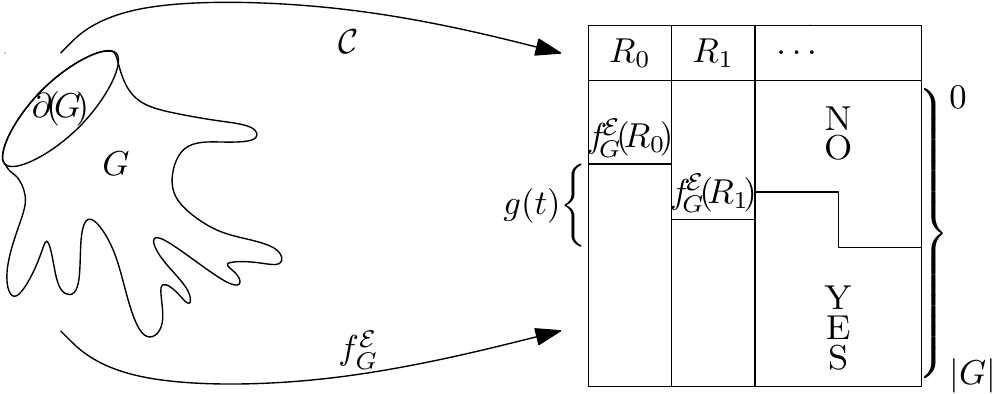}
\caption{Schematic illustration of a $g$-confined encoding.\label{fig:illustrationEncoding}}
\end{figure}

\begin{runningExample} 
It is easy to observe that the encoder $\mc{E}_{\textsc{DS}}$ described above is $g$-confined for $g(t)=t$. Indeed, let $G$ be a $t$-boundaried graph (corresponding to the graph $G_B$ considered before) with $\Lambda(G)=I$. Consider an arbitrary encoding $R\in\mc{C}(I)$ and the encoding $R_0\in\mc{C}(I)$ satisfying $R_0(v)=0$ for every $v\in \partial(G)$. Let $S_0\subseteq V(G)$ be a minimum-sized partial dominating set satisfying $R_0$, i.e.,  such that $(G,S_0,R_0)\in L_{\mc{C}_{\textsc{DS}}}$. Observe that $S_0$ also satisfies $R$, i.e., $(G,S_0,R)\in L_{\mc{C}_{\textsc{DS}}}$. It then follows
that $f_{G}^{\mc{E}_{\textsc{DS}}}(R_{0})  = \max_{R} f_{G}^{\mc{E}_{\textsc{DS}}}(R)$.
Moreover, let $S\subseteq V(G)$ be a minimum-sized partial dominating set satisfying $R$, i.e., such that $(G,S,R) \in L_{\mc{C}_{DS}}$. Then note that $R_0$ is satisfied by the set $S \cup \partial(G)$, so we have that for every encoding $R$, $f_{G}^{\mc{E}_{\textsc{DS}}}(R)  +   |\partial(G)| \geq f_{G}^{\mc{E}_{\textsc{DS}}}(R_{0})$. It follows that $f_{G}^{\mc{E}_{\textsc{DS}}}(R_{0})  -   \min_{R}f_{G}^{\mc{E}_{\textsc{DS}}}(R)  \  \leqslant \ |\partial(G)| \leq t $, proving that the encoder is indeed $g$-confined. \end{runningExample}

For some problems and encoders, we may need to ``force'' the confinement of an encoder $\mc{E}$ that may not be confined according to Definition~\ref{def:thin}, while still preserving its usefulness for dynamic programming, in the sense that no relevant information is removed from the tables  (for example, see the encoder for \textsc{$r$-Scattered Set} in Subsection~\ref{sec:Description-encoder-rScatSet}). To this end, given a function $g: \mathds{N} \to \mathds{N}$, we define the function $f_{G}^{ \mathcal{E},g}: \mathcal{C}(I) \to \mathds{N} \cup \{+\infty\}$ as \vspace{-.5cm}

\begin{equation}\label{eq:fEmin-g}
f_{G}^{ \mathcal{E},g}(R)\ =\ \left\{\begin{array}{lll}
& +\infty,\ \ & \mbox{\ \ if $f_{G}^{ \mathcal{E}}(R) - g(t)> \min_{R \in \mathcal{C}(I)}f_{G}^{ \mathcal{E}}(R)$}\\
& f_{G}^{ \mathcal{E}}(R), &\ \ \mbox{otherwise.}\\
\end{array}\right.
\end{equation}

Intuitively, one shall think as the function $f_{G}^{ \mathcal{E},g}$ as a ``compressed'' version of the function $f_{G}^{ \mathcal{E}}$, which stores only the values that are useful for performing dynamic programming.

For encoders $\mc{E}'=(\mc{C}',L_{\mc{C}'})$  associated with maximization problems, given a function $g: \mathds{N} \to \mathds{N}$, we define the function $f_{G}^{ \mathcal{E}',g}: \mathcal{C}(I) \to \mathds{N} \cup \{-\infty\}$ as

\begin{equation}\label{eq:fEmin-g}
f_{G}^{ \mathcal{E}',g}(R)\ =\ \left\{\begin{array}{lll}
& -\infty,\ \ & \mbox{\ \ if $f_{G}^{ \mathcal{E}'}(R) + g(t) < \max_{R \in \mathcal{C}(I)}f_{G}^{ \mathcal{E}'}(R)$}\\
& f_{G}^{ \mathcal{E}'}(R), &\ \ \mbox{otherwise.}\\
\end{array}\right.
\end{equation}


\subsection{Equivalence relations and representatives}
\label{sec:equivalencerelations}

An encoder $\mathcal{E}$ together with a function $g: \mathds{N} \to \mathds{N}$
define an equivalence relation $\sim^*_{\mathcal{E},g,t}$ on $t$-boundaried graphs as follows. (In fact, in our applications we will use only this equivalence relation on graphs in $\mathcal{F}_t$, but for technical reasons we need to define it on general $t$-boundaried graphs.)

\begin{definition}[Equivalence relations $\sim^*_{\mathcal{E},g,t}$ and $\sim_{\mathcal{E},g,t}$]\label{def:equivalenceRelation}
Let $\mathcal{E}$ be an encoder, let $g:\mathds{N} \to \mathds{N}$, and let $G_1,G_2 \in \mathcal{B}_t$. We say that $G_1 \sim^*_{\mathcal{E},g,t} G_2$ if and only if $\Lambda(G_1)=\Lambda(G_2)=: I$ and  there exists an integer $c$, depending only on $G_1$ and $G_2$,  such that for every $\mathcal{C}$-encoding $R \in \mathcal{C}(I)$ it holds that
\begin{equation}\label{eq:equivTruncation}
f_{G_1}^{ \mathcal{E},g}(R)\ =\ f_{G_2}^{ \mathcal{E},g}(R) + c.
\end{equation}
If we restrict the graphs $G_1,G_2$ to belong to $\mathcal{F}_t$, then the corresponding equivalence relation, which is a restriction of $\sim^*_{\mathcal{E},g,t}$, is denoted  by $\sim_{\mathcal{E},g,t}$.
\end{definition}

Note that if there exists $R \in \mathcal{C}(I)$ such that $f_{G_1}^{ \mathcal{E},g}(R) \notin  \{-\infty, + \infty\}$, then  the integer $c$ satisfying Equation~(\ref{eq:equivTruncation}) is unique, otherwise every integer $c$ satisfies Equation~(\ref{eq:equivTruncation}).
 We define the following function $\Delta_{\mathcal{E},g,t}: \mathcal{B}_t \times \mathcal{B}_t \to \mathds{Z}$, which is called, following the terminology from Bodlaender \emph{et al}.\ci{BFL+09},  the \emph{transposition function} for the equivalence relation $\sim^*_{\mathcal{E},g,t}$.
\begin{equation}
\Delta_{\mathcal{E},g,t}(G_1,G_2) \ =\ \left\{
\begin{array}{lll}
& c, \ \ & \mbox{if}\  G_1\sim^*_{\mathcal{E},g,t} G_2 \ \mbox{and Eq.~(\ref{eq:equivTruncation}) holds for a unique integer $c$;}\\
& 0, \ \ & \mbox{if}\  G_1\sim^*_{\mathcal{E},g,t} G_2 \ \mbox{and Eq.~(\ref{eq:equivTruncation}) holds for every integer; and}\\
&  & \mbox{undefined otherwise}
\end{array}
\right.
\end{equation}

Note that we can consider the restriction of the function $\Delta_{\mathcal{E},g,t}$ to couples of graphs in $\mc{F}_t$, defined by using the restricted equivalence relation $\sim_{\mathcal{E},g,t}$.

If we are dealing with a problem defined on a graph class $\mathcal{G}$, the protrusion replacement rule has to preserve the class $\mc{G}$, as otherwise we would obtain a \emph{bikernel} instead of a kernel. That is, we need to make sure that, when replacing a graph in $\mc{B}_t \cap \mc{G}$ or in $\mc{F}_t \cap \mc{G}$ with one of its representatives, we do not produce a graph that does not belong to $\mathcal{G}$ anymore. To this end, we define an equivalence relation $\sim^*_{\mathcal{E},g,t,\mathcal{G}}$ (resp. $\sim_{\mathcal{E},g,t,\mathcal{G}}$) on graphs in $\mc{B}_t \cap \mc{G}$ (resp. $\mc{F}_t \cap \mc{G}$), which refines the equivalence relation $\sim^*_{\mathcal{E},g,t}$ (resp. $\sim_{\mathcal{E},g,t}$) of Definition~\ref{def:equivalenceRelation}.

\begin{definition}[Equivalence relations $\sim^*_{\mathcal{E},g,t,\mathcal{G}}$ and $\sim_{\mathcal{E},g,t,\mathcal{G}}$]\label{def:equivalenceRelation2}
Let $\mc{G}$ be a class of graphs and let $G_1,G_2 \in \mathcal{B}_t \cap \mc{G}$.
\begin{itemize}
\item[(i)] 
$G_1 \sim_{\mc{G},t} G_2$ if and only if for any graph $H \in \mathcal{B}_t$, $G_1 \oplus H \in \mathcal{G}$ if and only if $G_2 \oplus H \in \mathcal{G}$.
\item[(ii)] $G_1 \sim^*_{\mathcal{E},g,t,\mathcal{G}} G_2$ if and only if $G_1 \sim^*_{\mathcal{E},g,t} G_2$ and $G_1 \sim_{\mc{G},t} G_2$.
\end{itemize}
If we restrict the graphs $G_1,G_2$ to belong to $\mathcal{F}_t$ (but still $H \in \mathcal{B}_t$), then the corresponding equivalence relation, which is a restriction of $\sim^*_{\mathcal{E},g,t,\mathcal{G}}$, is denoted  by $\sim_{\mathcal{E},g,t,\mathcal{G}}$.
\end{definition}

It is well-known by B\"{u}chi's theorem that regular languages are precisely those definable in Monadic Second Order logic (MSO logic). By Myhill-Nerode's theorem, it follows that if the membership in a graph class $\mc{G}$ can be expressed in MSO logic, then the equivalence relation $\sim_{\mc{G},t}$ has a finite number of equivalence classes (see for instance~\cite{FlGr06,DF99}). However, we do not have in general an explicit upper bound on the number of equivalence classes of $\sim_{\mc{G},t}$, henceforth denoted by $r_{\mc{G},t}$. Fortunately, in the context of our applications in Sections~\ref{sec:rDomSet},~\ref{sec:rScatSet}, and~\ref{sec:PlanarFDeletion}, where $\mc{G}$ will be a class of graphs that exclude some fixed graph on $h$ vertices as a (topological) minor\footnote{A particular case of the classes of graphs whose membership can be expressed in MSO logic. We would like to stress here that we rely on the expressibility of the \emph{graph class} $\mc{G}$ in MSO logic, whereas in previous work~\cite{BFL+09,FLST10,KLP+12} what is used in the expressibility in CMSO logic of the \emph{problems} defined on a graph class.}, this will always be possible, and in this case it holds that  $r_{\mc{G},t} \leq 2^{t \log t} \cdot h^t \cdot 2^{h^2}$.


For an encoder $\mc{E}=(\mc{C},L_{\mc{C}})$, we let $s_{\mc{E}}(t)  :=  \max_{I \subseteq \{1,\ldots,t\}} |\mc{C}(I)|$, where $|\mc{C}(I)|$ denotes the number of  $\mathcal{C}$-encodings in $\mc{C}(I)$. The following lemma gives an upper bound on the number of equivalence classes of $\sim^*_{\mathcal{E},g,t,\mathcal{G}}$, which depends also on  $r_{\mc{G},t}$.


\begin{lemma}\label{lem:finite}  Let $\mathcal{G}$ be a graph class whose membership can be expressed in MSO logic. For any encoder $\mathcal{E}$, any function $g:\mathds{N} \to \mathds{N}$, and any positive integer $t$, the equivalence relation $\sim^*_{\mathcal{E},g,t,\mathcal{G}}$ has finite index. More precisely, the number of equivalence classes of $\sim^*_{\mathcal{E},g,t,\mathcal{G}}$ is at most $r(\mathcal{E},g,t,\mathcal{G}):={(g(t)+2)^{s_{\mc{E}}(t)}} \cdot 2^t \cdot r_{\mc{G},t}$. 
In particular, the number of equivalence classes of $\sim_{\mathcal{E},g,t,\mathcal{G}}$ is at most $r(\mathcal{E},g,t,\mathcal{G})$ as well.
\end{lemma}
\begin{proof}
Let us first show that the equivalence relation $\sim^*_{\mathcal{E},g,t}$ has finite index. Indeed, let $I \subseteq \{1,\ldots,t\}$. By definition, we have that for any graph $G \in \mc{B}_t$ with $\Lambda(G)=I$, the function $f_{G}^{ \mathcal{E},g}$ can take at most $g(t)+2$ distinct values ($g(t)+1$ finite values and possibly the value $+\infty$).  Therefore, it follows that the number of equivalence classes of  $\sim^*_{\mathcal{E},g,t}$ containing all graphs $G$ in $\mc{B}_t$ with $\Lambda(G)=I$ is at most  ${(g(t)+2)^{|\mathcal{C}(I)|}}$. As the number of subsets of  $\{1,\ldots,t\}$ is $2^t$, we deduce that the overall number of equivalence classes of $\sim^*_{\mathcal{E},g,t}$ is at most ${(g(t)+2)^{s_{\mc{E}}(t)}} \cdot 2^t$. Finally, since the equivalence relation $\sim^*_{\mathcal{E},g,t,\mathcal{G}}$ is the Cartesian product of the equivalence relations $\sim^*_{\mathcal{E},g,t}$ and $\sim_{\mc{G},t}$, the result follows from the fact that $\mathcal{G}$ can be expressed in MSO logic. \end{proof}

%

In order for an encoding $\mathcal{E}$ and a function $g$ to be useful for performing dynamic programming on graphs in $\mathcal{F}_t$ that belong to a graph class $\mc{G}$ (recall that this is our final objective), we introduce the following definition, which captures the natural fact that the tables of a dynamic programming algorithm should depend exclusively on the tables of the descendants in a rooted tree decomposition.  Before moving to the definition, we note that given a graph $G \in \mathcal{F}_t$ and a rooted tree decomposition $(T,\mathcal{X})$ of $G$  of width at most $t-1$ such that $\partial(G)$ is contained in the root-bag of  $(T,\mathcal{X})$, the labels of $\partial(G)$ can be propagated in a natural way to all bags of $(T,\mathcal{X})$ by introducing, removing, and shifting labels appropriately. Therefore, for any node $x$ of $T$, the graph $G_x$ can be naturally seen as a graph in $\mc{F}_t$. (A brief discussion can be found in the proof of Lemma~\ref{lem:compute}, and we refer to~\cite{BFL+09} for more details.)

Again, for technical reasons (namely, for the proof of Lemma~\ref{lem:characterize}), we need to state the definition below for graphs in $\mc{B}_t$, even if we will only use it for graphs in $\mc{F}_t$.

\begin{definition}[DP-friendly equivalence relation]\label{def:DPfriendly}
 An equivalence relation $\sim^*_{\mathcal{E},g,t,\mathcal{G}}$ is \emph{DP-friendly} if for any graph $G \in \mc{B}_t$ with $\partial(G)=A$ and any separator $B \subseteq V(G)$ with $|B| \leq t$, the following holds: let $G_B$ be any collection of connected components of $G - B$ such that $A \cap V(G_B) \subseteq B$. Considering $G_B$ as a $t$-boundaried graph with boundary $B$, let $G'$ be the $t$-boundaried graph with $\partial(G')=A$ obtained from $G$ by replacing the subgraph $G_B$ with a $t$-boundaried graph $G_B'$ such that $G_B \sim^*_{\mathcal{E},g,t,\mathcal{G}} G_B'$. Then $G'$ satisfies the following conditions: 
\begin{itemize}
 \item[(i)] $G \sim^*_{\mathcal{E},g,t,\mc{G}} G'$; and
 \item[(ii)] $\Delta_{\mathcal{E},g,t}(G,G') = \Delta_{\mathcal{E},g,t}(G_B,G_B')$.
\end{itemize}\vspace{-.1cm}
\end{definition}

Note that if an equivalence relation $\sim^*_{\mathcal{E},g,t,\mathcal{G}}$ is DP-friendly, then by definition its restriction $\sim_{\mathcal{E},g,t,\mathcal{G}}$ to graphs in $\mc{F}_t$ is DP-friendly as well.

We would like to note that in the above definition we have used the notation $G_B$ because in all applications the subgraph to be replaced will correspond to a rooted subtree in a tree decomposition of a graph $G$. With this in mind, the condition $A \cap V(G_B) \subseteq B$ in Definition~\ref{def:DPfriendly} corresponds to the fact that the boundary $A$ will correspond in the applications to the vertices in the root-bag of a rooted tree decomposition of $G$.

In Definition~\ref{def:DPfriendly}, as well as in the remainder of the article, when we \emph{replace} the graph  $G_B$ with the graph $G_B'$,  we do \emph{not} remove from $G$ any of the edges with both endvertices in the boundary of $G_B$. That is, $G'= (G - (V(G_x)-\partial(V(G_B))))\oplus G_B'$.

Recall that for the protrusion replacement to be valid for a problem $\Pi$, the equivalence relation $\sim_{\mathcal{E},g,t,\mathcal{G}}$ needs to be a refinement of the canonical equivalence relation $\equiv_{\Pi,t}$ (note that this implies, in particular, that if $\sim_{\mathcal{E},g,t,\mathcal{G}}$ has finite index, then $\Pi$ has FII). The next lemma states a sufficient condition for this property, and furthermore it gives the value of the transposition constant $\Delta_{\Pi,t}(G_1,G_2)$, which will be needed in order to update the parameter after the replacement.

\begin{lemma}\label{lem:characterize}  Let $\Pi$ be a vertex-certifiable problem.  If $\mathcal{E}$ is a $\Pi$-encoder and $\sim^*_{\mathcal{E},g,t,\mathcal{G}}$ is a DP-friendly equivalence relation, then for any two graphs $G_1,G_2 \in \mc{B}_t$ such that $G_1 \sim^*_{\mathcal{E},g,t,\mc{G}} G_2$, it holds that $G_1 \equiv_{\Pi,t} G_2$ and $\Delta_{\Pi,t}(G_1,G_2) = \Delta_{\mathcal{E},g,t}(G_1,G_2)$. In particular, if $\mathcal{E}$ is a $\Pi$-encoder and $\sim^*_{\mathcal{E},g,t,\mathcal{G}}$ is DP-friendly, then for any two graphs $G_1,G_2 \in \mc{F}_t$ such that $G_1 \sim_{\mathcal{E},g,t,\mc{G}} G_2$, it holds that $G_1 \equiv_{\Pi,t} G_2$ and $\Delta_{\Pi,t}(G_1,G_2) = \Delta_{\mathcal{E},g,t}(G_1,G_2)$.
\end{lemma}

\begin{proof}
Assume without loss of generality that $\Pi$ is a minimization problem, and let $\mathcal{E} = (\mc{C},L_{\mc{C}})$.
We need to prove that for any $t$-boundaried graph $H$ and any integer $k$, $(G_1 \oplus H, k ) \in \Pi$ if and only if $(G_2 \oplus H, k + \Delta_{\mathcal{E},g,t}(G_1,G_2)) \in \Pi$. Suppose that  $(G_1 \oplus H, k ) \in \Pi$ (by symmetry the same arguments apply starting with $G_2$). This means that there exists $S_1 \subseteq V(G_1 \oplus H)$ with $|S_1| \leq k$ such that $(G_1 \oplus H,S_1) \in L_{\Pi}$. And since $G_1 \oplus H$ is a $0$-boundaried graph and $\mathcal{E}$ is a $\Pi$-encoder, we have that $(G_1 \oplus H,S_1,R_{\emptyset}) \in L_{\mc{C}}$, where $\mc{C}(\emptyset) = \{R_{\emptyset} \}$.  This implies that
\begin{equation}\label{eq:LemmaCharacterize}
f_{G_1 \oplus H}^{ \mathcal{E}}(R_{\emptyset})\ \leq\ |S_1|\ \leq\ k.
\end{equation}
As $\sim^*_{\mathcal{E},g,t,\mathcal{G}}$ is DP-friendly and $G_1 \sim^*_{\mathcal{E},g,t,\mc{G}} G_2$, it follows that  $(G_1 \oplus H) \sim^*_{\mathcal{E},g,t,\mc{G}}(G_2 \oplus H)$
and that
 $\Delta_{\mathcal{E},g,t}(G_1 \oplus H,G_2 \oplus H) = \Delta_{\mathcal{E},g,t}(G_1,G_2)$. Since $G_2 \oplus H$ is also a $0$-boundaried graph, the latter property and Equation~(\ref{eq:LemmaCharacterize}) imply that
 \begin{equation}\label{eq:LemmaCharacterize2}
 f_{G_2 \oplus H}^{ \mathcal{E}}(R_{\emptyset})\ =\ f_{G_1 \oplus H}^{ \mathcal{E}}(R_{\emptyset}) + \Delta_{\mathcal{E},g,t}(G_1,G_2)\  \leq\ k + \Delta_{\mathcal{E},g,t}(G_1,G_2).
 \end{equation}
 From Equation~(\ref{eq:LemmaCharacterize2}) it follows that there exists  $S_2 \subseteq V(G_2 \oplus H)$ with $|S_2| \leq k + \Delta_{\mathcal{E},g,t}(G_1,G_2)$ such that $(G_2 \oplus H,S_2,R_{\emptyset}) \in L_{\mathcal{C}}$. Since $G_2 \oplus H$ is a $0$-boundaried graph and $\mathcal{E}$ is a $\Pi$-encoder, this implies that $(G_2 \oplus H, S_2) \in L_{\Pi}$, which in turn implies that $(G_2 \oplus H, k + \Delta_{\mathcal{E},g,t}(G_1,G_2)) \in \Pi$, as we wanted to prove.
 \end{proof}

Note that, in particular, Lemma~\ref{lem:characterize} implies that under the same hypothesis, for graphs in $\mc{F}_t$ the equivalence relation $\sim_{\mathcal{E},g,t,\mathcal{G}}$ refines the canonical equivalence relation $\equiv_{\Pi,t}$.

In the following, we will only deal with equivalence relations $\sim_{\mathcal{E},g,t,\mathcal{G}}$ defined on graphs in $\mc{F}_t$, and therefore we will only use this particular case of Lemma~\ref{lem:characterize}. The reason why we restrict ourselves to graphs in $\mc{F}_t$ is that, while a DP-friendly equivalence relation refines the canonical one for all graphs in $\mc{B}_t$ (Lemma~\ref{lem:characterize}), we need {\sl bounded treewidth} in order to bound the {\sl size} of the progressive representatives (Lemma~\ref{lem:finiteSize}) and  to explicitly {\sl compute} these representatives for performing the replacement (Lemma~\ref{lem:compute}).

The following definition will be important to guarantee that, when applying our protrusion replacement rule, the parameter of the problem under consideration does not increase.

\begin{definition}[Progressive representatives of  $\sim_{\mathcal{E},g,t,\mathcal{G}}$]\label{def:progressive}
Let $\mathfrak{C}$ be some equivalence class of $\sim_{\mathcal{E},g,t,\mathcal{G}}$ and let $G \in \mathfrak{C}$. We say that $G$ is a \emph{progressive representative} of $\mathfrak{C}$ if for any graph $G' \in \mathfrak{C}$ it holds that $\Delta_{\mathcal{E},g,t}(G,G') \leq 0$.
\end{definition}

In the next lemma we provide an upper bound on the size of a smallest {\sl progressive} representative of any equivalence class of $\sim_{\mathcal{E},g,t,\mathcal{G}}$.



\begin{lemma}\label{lem:finiteSize}
Let $\mathcal{G}$ be a graph class whose membership can be expressed in MSO logic. For any encoder $\mathcal{E}$, any function $g:\mathds{N} \to \mathds{N}$, and any $t \in \mathds{N}$ such that $\sim^*_{\mathcal{E},g,t,\mathcal{G}}$ is DP-friendly, every equivalence class of $\sim_{\mathcal{E},g,t,\mathcal{G}}$ has a progressive representative of size at most $b(\mathcal{E},g,t,\mathcal{G}):= 2^{r(\mathcal{E},g,t,\mathcal{G})+1} \cdot t$, where $r(\mathcal{E},g,t,\mathcal{G})$ is the function defined in Lemma~\ref{lem:finite}.
\end{lemma}
\begin{proof} Let $\mathfrak{C}$ be an arbitrary equivalence class of $\sim_{\mathcal{E},g,t,\mathcal{G}}$, and we want to prove that there exists in $\mathfrak{C}$  a progressive representative of the desired size.
Let us first argue that $\mathfrak{C}$ contains some progressive representative. We construct an (infinite) directed graph $D_{\mathfrak{C}}$ as follows. There is a vertex in $D_{\mathfrak{C}}$ for every graph in  $\mathfrak{C}$, and for any two vertices $v_1,v_2 \in V(D_{\mathfrak{C}})$, corresponding to two graphs $G_1,G_2 \in \mathfrak{C}$ respectively, there is an arc from $v_1$ to $v_2$ if and only if $\Delta_{\mathcal{E},g,t}(G_1,G_2) > 0$.  We want to prove that $D_{\mathfrak{C}}$ has a sink, that is, a vertex with no outgoing arc, which by construction is equivalent to the existence of a progressive representative in $\mathfrak{C}$. Indeed, let $v$ be an arbitrary vertex of $D_{\mathfrak{C}}$, and grow greedily a directed path $P$ starting from $v$. Because of the transitivity of the equivalence relation $\sim_{\mathcal{E},g,t,\mathcal{G}}$ and by construction of $D_{\mathfrak{C}}$, it follows that $D_{\mathfrak{C}}$ does not contain any finite cycle, so $P$ cannot visit vertex $v$ again. On the other hand, since the function $f_{G}^{ \mathcal{E}}$ takes only positive values (except possibly for the value $-\infty$), it follows that there are no arbitrarily long directed paths in $D_{\mathfrak{C}}$ starting from any fixed vertex, so in particular the path $P$ must be finite, and therefore the last vertex in $P$ is necessarily a sink. (Note that for any two graphs $G_1,G_2 \in \mathfrak{C}$ such that their corresponding vertices $v_1,v_2 \in V(D_{\mathfrak{C}})$ are sinks, it holds by construction of $D_{\mathfrak{C}}$ that $\Delta_{\mathcal{E},g,t}(G_1,G_2) = 0$.)

Now let $G \in \mc{F}_t \cap \mc{G}$ be a progressive representative of $\mathfrak{C}$ with minimum number of vertices. We claim that $G$ has size at most $2^{r(\mathcal{E},g,t,\mathcal{G})+1} \cdot t$. (We would like to stress that at this stage we only need to care about the {\sl existence} of such representative $G$, and not about how to {\sl compute} it.) Indeed, let $(T,\mc{X})$ be a nice rooted tree decomposition of $G$ of width at most $t-1$ such that $\partial(G)$ is contained in the root-bag (such a nice tree decomposition exists by~\cite{Klo94}), and let $r$ be the root of $T$.

We first claim that for any node $x$ of $T$, the graph $G_x$ is a progressive representative of its equivalence class with respect to $\sim_{\mathcal{E},g,t,\mathcal{G}}$, namely $\mathfrak{A}$. Indeed, assume that it is not the case, and let $H$ be a progressive representative of $\mathfrak{A}$, which exists by the discussion in the first paragraph of the proof. Since $H$ is progressive and $G_x$ is not, $\Delta_{\mathcal{E},g,t}(H,G_x) < 0$. Let $G_H$ be the graph obtained from $G$ by replacing $G_x$ with $H$. Since  $\sim^*_{\mathcal{E},g,t,\mathcal{G}}$ is DP-friendly, it follows that  $G \sim_{\mathcal{E},g,t,\mc{G}} G_H$ and that $\Delta_{\mathcal{E},g,t}(G_H,G) = \Delta_{\mathcal{E},g,t}(H,G_x) < 0$, contradicting the fact that $G$ is a progressive representative of the equivalence class $\mathfrak{C}$.


We now claim that for any two nodes $x,y \in V(T)$ lying on a path from $r$ to a leaf of $T$, it holds that $G_x \nsim_{\mathcal{E},g,t,\mathcal{G}} G_y$. Indeed, assume for contradiction that there are two nodes $x,y \in V(T)$ lying on a path from $r$ to a leaf of $T$ such that $G_x \sim_{\mathcal{E},g,t,\mc{G}} G_y$. Let $\mathfrak{A}$ be the equivalence class of $G_x$ and $G_y$ with respect to $\sim_{\mathcal{E},g,t,\mathcal{G}}$. By the previous claim, it follows that both $G_x$ and $G_y$ are progressive representatives of $\mathfrak{A}$, and therefore it holds that $\Delta_{\mathcal{E},g,t}(G_y,G_x) = 0$. Suppose without loss of generality that $G_y \subsetneq G_x$ (that is, $G_y$ is a strict subgraph of $G_x$), and let $G'$ be the graph obtained from $G$  by replacing $G_x$ with $G_y$. Again, since $\sim^*_{\mathcal{E},g,t,\mathcal{G}}$ is DP-friendly, it follows that  $G \sim_{\mathcal{E},g,t,\mc{G}} G'$ and that $\Delta_{\mathcal{E},g,t}(G',G) = \Delta_{\mathcal{E},g,t}(G_y,G_x) = 0$. Therefore, $G'$ is a progressive representative of $\mathfrak{C}$ with $|V(G')| < |V(G)|$, contradicting the minimality of $|V(G)|$.

Finally, since for any two nodes $x,y \in V(T)$ lying on a path from $r$ to a leaf of $T$ we have that $G_x \nsim_{\mathcal{E},g,t,\mathcal{G}} G_y$, it follows that the height of $T$ is at most the number of equivalence classes of $\sim_{\mathcal{E},g,t,\mathcal{G}}$, which is at most $r(\mathcal{E},g,t,\mathcal{G})$ by Lemma~\ref{lem:finite}. Since $T$ is a binary tree, we have that $|V(T)| \leq 2^{r(\mathcal{E},g,t,\mathcal{G})+1} - 1$. Finally, since $|V(G)| \leq |V(T)| \cdot t$, it follows that $|V(G)| \leq 2^{r(\mathcal{E},g,t,\mathcal{G})+1} \cdot t$, as we wanted to prove. \end{proof}


The next lemma states that if one is given an upper bound on the size of the progressive representatives of an equivalence relation defined on $t$-protrusions  (that is, on graphs in $\mc{F}_t$)\footnote{Note that we slightly abuse notation when identifying $t$-protrusions and graphs  in $\mc{F}_t$, as protrusions are defined as subsets of vertices of a graph. Nevertheless, this will not cause any confusion.}, then a {\sl small} progressive representative of a $t$-protrusion can be explicitly calculated in linear time. In other words, it provides a generic and constructive way to perform a dynamic programming procedure to replace protrusions, without needing to deal with the particularities of each encoder in order to compute the tables. Its proof uses some ideas taken from~\cite{BFL+09,FLST10}.

\begin{lemma}
\label{lem:compute}
Let $\mathcal{G}$ be a graph class, let $\mathcal{E}$ be an encoder, let $g:\mathds{N} \to \mathds{N}$, and let $t \in \mathds{N}$ such that $\sim^*_{\mathcal{E},g,t,\mathcal{G}}$ is DP-friendly. Assume that we are given an upper bound $b$ on the size of a smallest progressive representative
of any equivalence class of $\sim_{\mathcal{E},g,t,\mathcal{G}}$, with $b \geq t$.
Then,  given an $n$-vertex $t$-protrusion $G$ inside some graph, we can output in time $O(n)$ a $t$-protrusion $H$ inside the same graph of size at most $b$ such that $G \sim_{\mathcal{E},g,t,\mathcal{G}} H$ and the corresponding transposition constant $\Delta_{\mathcal{E},g,t}(H,G)$ with $\Delta_{\mathcal{E},g,t}(H,G) \leq 0$, where the hidden constant in the ``$O$'' notation depends only on $\mathcal{E},g,b,\mathcal{G}$, and $t$.
\end{lemma}
\begin{proof} Let $\mc{E} = (\mc{C},L_{\mc{C}})$ be the given encoder. We start by generating a repository $\mathfrak{R}$ containing all the graphs in $\mc{F}_t$ with at most $b+1$ vertices. Such a set of graphs, as well as a rooted nice tree decomposition of width at most $t-1$ of each of them, can be clearly generated in time depending only on $b$ and $t$. By assumption, the size of a smallest progressive representative
of any equivalence class of $\sim_{\mathcal{E},g,t,\mathcal{G}}$ is at most $b$, so  $\mathfrak{R}$ contains a progressive representative of any equivalence class of $\sim_{\mathcal{E},g,t,\mathcal{G}}$ with at most $b$ vertices. We now partition the graphs in $\mathfrak{R}$ into equivalence classes of $\sim_{\mathcal{E},g,t,\mathcal{G}}$ as follows. For each graph $H \in \mathfrak{R}$ and each $\mc{C}$-encoding $R \in \mathcal{C}(\Lambda(G))$, as $L_\mathcal{C}$ is a computable language, we can
compute the value $f_{G}^{ \mathcal{E},g}(R)$ in time depending only on $\mc{E},g,t,$ and $b$. Therefore, for any two graphs $H_1,H_2 \in \mathfrak{R}$, we can decide in time depending only on $\mc{E},g,t,b$, and $\mc{G}$ whether $H_1 \sim_{\mathcal{E},g,t,\mathcal{G}} H_2$, and if this is the case, we can compute the transposition constant $\Delta_{\mathcal{E},g,t}(H_1,H_2)$ within the same running time.

Given a $t$-protrusion $G$ on $n$ vertices with boundary $\partial(G)$, we first compute a rooted nice tree decomposition $(T,\mc{X})$ of $G$ such that $\partial(G)$ is contained in the root bag
in time $f(t) \cdot n$, by using the linear-time algorithm of Bodlaender~\cite{Bod96,Klo94}. Such a $t$-protrusion $G$ equipped with a tree decomposition can be naturally seen as a graph in $\mc{F}_t$ by assigning distinct labels from $\{1,\ldots,t\}$ to the vertices in the root-bag. These labels from $\{1,\ldots,t\}$ can be transferred to the vertices in all the bags of $(T,\mc{X})$ by performing a standard shifting procedure when a vertex is introduced or removed from the nice tree decomposition (see~\cite{BFL+09} for more details). Therefore, each node $x \in V(T)$ defines in a natural way a graph  $G_x \subseteq G$ in $\mc{F}_t$ with its associated rooted nice tree decomposition. Let us now proceed to the description of the replacement algorithm.

We process the bags of $(T,\mc{X})$ in a bottom-up way until we encounter the first node $x$ in $V(T)$ such that $|V(G_x)|=b+1$.  (Note that as $(T,\mc{X})$ is a nice tree decomposition, when processing the bags in a bottom-up way, at most one new vertex is introduced at every step, and recall that by hypothesis $t \leq b$.) Let $\mathfrak{C}$ be the equivalence class of $G_x$ according to $\sim_{\mathcal{E},g,t,\mathcal{G}}$. As $|V(G_x)|=b+1$, the graph $G_x$ is contained in the repository $\mathfrak{R}$, so in constant time we can find in  $\mathfrak{R}$ a progressive representative $F$ of $\mathfrak{C}$ with at most $b$ vertices and the corresponding transposition constant $\Delta_{\mathcal{E},g,t}(F,G_x) \leq 0$, where the inequality holds because $F$ is a progressive representative. Let $G'$ be the graph obtained from $G$ by replacing $G_x$ with $F$, so we have that $|V(G')| < |V(G)|$. (Note that this replacement operation directly yields a rooted nice tree decomposition of width at most $t-1$ of $G'$.) Since $\sim^*_{\mathcal{E},g,t,\mathcal{G}}$ is DP-friendly, it follows that  $G \sim_{\mathcal{E},g,t,\mc{G}} G'$ and that $\Delta_{\mathcal{E},g,t}(G',G) = \Delta_{\mathcal{E},g,t}(F,G_x) \leq 0$.

We recursively apply this replacement procedure on the resulting graph until we eventually obtain a $t$-protrusion $H$ with at most $b$ vertices such that $G \sim_{\mathcal{E},g,t,\mathcal{G}} H$. The corresponding transposition constant $\Delta_{\mathcal{E},g,t}(H,G)$ can be easily computed by summing up all the transposition constants given by each of the performed replacements. Since each of these replacements introduces a progressive representative, we have that $\Delta_{\mathcal{E},g,t}(H,G) \leq 0$. As we can assume that the total number of nodes in a nice tree decomposition of $G$ is $O(n)$~\cite[Lemma 13.1.2]{Klo94}, the overall running time of the algorithm is $O(n)$, where the constant hidden in the ``$O$'' notation depends indeed exclusively on $\mathcal{E},g,b,\mathcal{G}$, and $t$. \end{proof}

\subsection{Explicit protrusion replacer}
\label{sec:explicitprotrusionreplacer}

We are now ready to piece everything together and state our main technical result, which can be interpreted as a generic {\sl constructive} way of performing protrusion replacement with {\sl explicit} size bounds. For our algorithms to be fully constructive, we restrict $\mathcal{G}$ to be the class of graphs that exclude some fixed graph $H$ as a (topological) minor.


\begin{theorem}\label{thm:protrusionReplacement}
Let $H$ be a fixed graph and let $\mathcal{G}$ be the class of graphs that exclude $H$ as a (topological) minor.
Let $\Pi$ be a vertex-certifiable parameterized graph problem defined on $\mathcal{G}$, and suppose that we are given a $\Pi$-encoder $\mathcal{E}$, a function $g: \mathds{N} \to \mathds{N}$, and an integer $t \in \mathds{N}$ such that $\sim^*_{\mathcal{E},g,t,\mathcal{G}}$ is DP-friendly. Then, given an input graph $(G,k)$ and a $t$-protrusion $Y$ in $G$, we can compute in time $O(|Y|)$  an equivalent instance $((G - (Y-\partial(Y)))\oplus Y',k')$, where $k' \leq k$ and $Y'$ is a $t$-protrusion with  $|Y'| \leq b(\mathcal{E},g,t,\mathcal{G})$, where $b(\mathcal{E},g,t,\mathcal{G})$ is the function defined in Lemma~\ref{lem:finiteSize}.
\end{theorem}

\begin{proof}
By Lemma~\ref{lem:finite}, the number of equivalence classes of the equivalence relation $\sim_{\mathcal{E},g,t,\mathcal{G}}$ is finite, and by Lemma~\ref{lem:finiteSize} the size of a smallest progressive representative of any equivalence class of $\sim_{\mathcal{E},g,t,\mathcal{G}}$ is at most $b(\mathcal{E},g,t,\mathcal{G})$. Therefore, we can apply Lemma~\ref{lem:compute} and deduce that, in time $O(|Y|)$, we can find a $t$-protrusion $Y'$ of size at most $b(\mathcal{E},g,t,\mathcal{G})$ such that $Y \sim_{\mathcal{E},g,t,\mathcal{G}} Y'$, and the corresponding transposition constant $\Delta_{\mathcal{E},g,t}(Y',Y)$ with $\Delta_{\mathcal{E},g,t}(Y',Y) \leq 0$. Since $\mathcal{E}$ is a  $\Pi$-encoder and $\sim^*_{\mathcal{E},g,t,\mathcal{G}}$ is DP-friendly, it follows from Lemma~\ref{lem:characterize} that $Y \equiv_{\Pi,t} Y'$ and that $\Delta_{\Pi,t}(Y',Y) = \Delta_{\mathcal{E},g,t}(Y',Y) \leq 0$. Therefore, if we set $k' := k +\Delta_{\Pi,t}(Y',Y)$, it follows that $(G,k)$ and $((G - (Y-\partial(Y)))\oplus Y',k')$ are indeed equivalent instances of $\Pi$ with $k' \leq k$ and   $|Y'| \leq b(\mathcal{E},g,t,\mathcal{G})$.
\end{proof}


The general recipe to use our framework on a parameterized problem $\Pi$ defined on a class of graphs $\mc{G}$ is as follows: one has just to define the tables to solve $\Pi$ via dynamic programming on graphs of bounded treewidth  (that is, the encoder $\mathcal{E}$ and the function $g$), check that $\mathcal{E}$ is a $\Pi$-encoder and that $\sim^*_{\mathcal{E},g,t,\mathcal{G}}$ is DP-friendly, and then Theorem~\ref{thm:protrusionReplacement} provides a linear-time algorithm that replaces large protrusions with graphs whose size is bounded by an explicit constant, and that updates the parameter of $\Pi$ accordingly. This protrusion replacer can then be used, for instance, whenever one is able to find a linear protrusion decomposition of the input graphs of $\Pi$ on some sparse graph class $\mc{G}$. In particular, Theorem~\ref{thm:protrusionReplacement} yields the following corollary.

\begin{corollary}\label{corollary:withProtrusionDecomposition}
Let $H$ be a fixed graph, and let $\mathcal{G}$ be the class of graphs that exclude $H$ as a (topological) minor. Let $\Pi$ be a vertex-certifiable parameterized graph problem on $\mathcal{G}$, and suppose that we are given a $\Pi$-encoder $\mathcal{E}$, a function $g: \mathds{N} \to \mathds{N}$, and an integer $t \in \mathds{N}$ such that $\sim^*_{\mathcal{E},g,t,\mathcal{G}}$ is DP-friendly. Then, given an instance $(G,k)$ of $\Pi$ together with an $(\alpha \cdot k, t)$-protrusion decomposition of $G$, we can construct a linear kernel for $\Pi$ of size at most $(1+b(\mathcal{E},g,t,\mathcal{G}))\cdot \alpha \cdot k$, where $b(\mathcal{E},g,t,\mathcal{G})$ is the function defined in Lemma~\ref{lem:finiteSize}.
\end{corollary}
\begin{proof} For $1 \leq i \leq \ell$, we apply the polynomial-time algorithm given by Theorem~\ref{thm:protrusionReplacement} to replace each $t$-protrusion $Y_i$ with a graph $Y_i'$ of size at most $b(\mathcal{E},g,t,\mathcal{G})$, and to update the parameter accordingly. In this way we obtain an equivalent instance $(G',k')$ such that $G' \in \mc{G}$, $k' \leq k$, and $|V(G')| \leq |Y_0| + \ell \cdot b(\mathcal{E},g,t,\mathcal{G}) \leq (1+b(\mathcal{E},g,t,\mathcal{G}))\alpha \cdot k$ . 
\end{proof}

Notice that once we fix the problem $\Pi$ and the class of graphs $\mathcal{G}$ where Corollary~\ref{corollary:withProtrusionDecomposition}
is applied, a kernel of size $c\cdot  k$ can be derived with a concrete upper bound for the value of $c$. Notice that such a bound depends on the problem $\Pi$ and the excluded (topological) minor $H$. In general, the bound
can be quite big as it depends on
the bound of Lemma~\ref{lem:finiteSize}, and this, in turn, depends on the bound of Lemma~\ref{lem:finite}.
However, as we see in the next section, more moderate estimations can be extracted for particular families of parameterized problems.

Before demonstrating the applicability of our framework by providing linear kernels for several families of problems on graphs excluding a fixed graph as a (topological) minor, we need another ingredient. Namely, the following result will be fundamental in order to find linear protrusion decompositions when a treewidth-modulator $X$ of the input graph $G$ is given, with $|X|=O(k)$. It is a consequence of~\cite[Lemma~3, Proposition~1, and Theorem~1]{KLP+12} and, loosely speaking, the algorithm consists in marking the bags of a tree decomposition of $G-X$ according to the number of neighbors in the set $X$. When the graph $G$ is restricted to exclude a fixed graph $H$ as a topological minor, it can be proved that the obtained protrusion decomposition is linear. All the details can be found in the full version of~\cite{KLP+12}.

\begin{theorem}[Kim \emph{et al}.~\cite{KLP+12}]\label{thm:protDec}
Let $c,t$ be two positive integers, let $H$ be an $h$-vertex graph, let $G$ be an $n$-vertex $H$-topological-minor-free graph, and let $k$ be a positive integer (typically corresponding to the parameter of a parameterized problem). If we are given a set $X \subseteq V(G)$ with $|X| \leq c \cdot k$ such that $\tw(G-X) \leq t$, then we can compute in time $O(n)$ an $((\alpha_{H} \cdot t \cdot c)\cdot k, 2t + h)$-protrusion decomposition of $G$, where $\alpha_{H}$ is a constant depending only on $H$, which is upper-bounded by $40 h^2 2 ^{5 h \log h}$. 
\end{theorem}

As mentioned in Subsection~\ref{sec:equivalencerelations}, if $\mc{G}$ is a graph class whose membership  can be expressed in MSO logic, then $\sim_{\mc{G},t}$ has a finite number of equivalence classes, namely $r_{\mc{G},t}$. In our applications, we will be only concerned with families of graphs $\mc{G}$ that exclude some fixed $h$-vertex graph $H$ as a (topological) minor. In this case, using standard dynamic programming techniques, it can be shown that $r_{\mc{G},t} \leq 2^{t \log t} \cdot h^t \cdot 2^{h^2}$. The details can be found in the encoder described in Subsection~\ref{sec:Description-encoder-PlanarFDeletion} for the \textsc{$\mc{F}$-Deletion} problem.


\section{An explicit linear kernel for \textsc{$r$-Dominating Set}}
\label{sec:rDomSet}
Let $r \geq 1$ be a fixed integer. We define the \textsc{$r$-Dominating Set} problem as follows.

\vspace{.4cm}
\begin{boxedminipage}{.95\textwidth}
\textsc{$r$-Dominating Set}\vspace{.1cm}

\begin{tabular}{ r l }
\textbf{~~~~Instance:} & A graph $G=(V,E)$  and a non-negative integer $k$. \\
\textbf{Parameter:} & The integer $k$.\\
\textbf{Question:} & Does $G$ have a set $S \subseteq V$ with $|S| \leq k$ and such that every vertex \\ &  in $V \setminus S$ is within distance at most $r$ from some vertex in $S$?\\
\end{tabular}
\end{boxedminipage}\vspace{.4cm}


\noindent For $r=1$, the \textsc{$r$-Dominating Set} problem corresponds to \textsc{Dominating Set}. Our encoder for \textsc{$r$-Dominating Set} is strongly inspired by the work of Demaine \emph{et al}.~\cite{DFHT05}, and it generalizes to one given for \textsc{Dominating Set} in the running example of Section~\ref{sec:framework}. The encoder for \textsc{$r$-Dominating Set}, which we call $\Eds =(\Cds,\Lcds )$, is described in Subsection~\ref{sec:Description-encoder-rDomSet}, and we show how to construct the linear kernel in Subsection~\ref{sec:Construction-kernel-rDomSet}.



\subsection{Description of the encoder}
\label{sec:Description-encoder-rDomSet}

Let $G$ be a boundaried graph with boundary $\partial(G)$ and let $I=\Lambda(G)$. The function \Cds maps $I$ to a set $\Cds(I)$ of \Cds-encodings. Each $R\in \Cds(I)$ maps $I$ to an $|I|$-tuple in $\{0,\uparrow 1, \downarrow 1, \dots, \uparrow r, \downarrow r \}^{|I|}$, and thus the coordinates of the tuple are in  one-to-one correspondence with the vertices of $\partial(G)$.  For a vertex $v\in\partial(G)$ we denote by $R(v)$ its coordinate in the $|I|$-tuple. For a subset $S$ of vertices of $G$, we say that $(G,S,R)$ belongs to the language $\Lcds$ (or that $S$ is a \emph{partial $r$-dominating set satisfying $R$}) if :

\begin{itemize}
\item[$\bullet$] for every vertex $v\in V(G)\setminus \partial(G)$, either $d_G(v,S)\leqslant r$ or there exists $ w \in \partial(G)$ such that $R(w) = \uparrow j$ and $d_G(v,w)+j\leqslant r$; and
\item[$\bullet$]  for every vertex $v\in \partial(G)$: $R(v)=0$ implies that $v\in S$, and  if $R(v)=\downarrow i$ for $1 \leq i \leq r$, then  there exists either $w\in S$ such that $d_G(v,w)\leqslant i$ or  $w\in \partial(G)$ such that $R(w) = \uparrow j$ and $d_G(v,w)+j\leqslant i$.
\end{itemize}

Observe that if $S$ is a partial $r$-dominating set satisfying $R$, then $S\cap\partial(G)$ contains the set of vertices $\{v\in\partial(G)\mid R(v)=0\}$, but it may also contain other vertices of $\partial(G)$. As the optimization version of \rdom is a minimization problem, by Equation~(\ref{eq:fEmin}) the function $f_G^{\Cds}(R)$ associates with a \Cds-encoding $R$ the minimum size of a partial $r$-dominating set $S$ satisfying $R$. By definition of $\Eds$, it is clear that 
\begin{equation}\label{eq:sizeEncoderrDomSet}
s_{\Eds }(t) \ \leq \ (2r+1)^t.
\end{equation}


\begin{lemma}\label{lem:rDomCharacterizes}
The encoder $\Eds$ is a $r\textsc{DS}$-encoder. Furthermore, if $\mc{G}$ is an arbitrary class of graphs and $g(t)=t$, then the equivalence relation $\sim^*_{\Eds,g,t,\mc{G}}$ is DP-friendly.
\end{lemma}

Before providing the proof of Lemma~\ref{lem:rDomCharacterizes}, we will first state a general  fact, which will be useful in order to prove that an encoder is DP-friendly.

\begin{fact}\label{fact:DPfriendly} To verify that an equivalence relation $\sim^*_{\mathcal{E},g,t,\mc{G}}$ satisfies Definition~\ref{def:DPfriendly}, property $(i)$ can be replaced with $G \sim^*_{\mathcal{E},g,t} G'$. That is, if  $G \sim_{\mathcal{E},g,t} G'$, then $G \sim_{\mathcal{E},g,t,\mc{G}} G'$ as well.
\end{fact}
\begin{proof}
Assume that $G \sim^*_{\mathcal{E},g,t} G'$, and we want to deduce that $G \sim^*_{\mathcal{E},g,t,\mc{G}} G'$, that is,  we just have to prove that $G \sim_{\mathcal{G},t} G'$. Let $H$ be a $t$-boundaried graph, and we need to prove that $G \oplus H \in \mathcal{G}$ if and only if $G' \oplus H \in \mathcal{G}$. Let $G^-$ such that $G = G_x \oplus  G^-$, and note that $G' = G'_x \oplus G^-$. We have that $G \oplus H = (G_x \oplus  G^-)  \oplus H = G_x \oplus  (G^-  \oplus H)$, and similarly we have that $G'_x \oplus  (G^-  \oplus H) = (G'_x \oplus  G^- ) \oplus H = G' \oplus H$. Since $G_x \sim_{\mathcal{G},t} G'_x$, it follows that  $G \oplus H = G_x \oplus  (G^-  \oplus H) \in \mc{G}$ if and only if $G'_x \oplus  (G^-  \oplus H)= G' \oplus H \in \mc{G}$.
\end{proof}

We will use the shortcut \textsc{$r$DS} for \textsc{$r$-Dominating Set}.\vspace{.4cm}

\begin{proofrDomSet} Let us first prove that $\Eds =(\Cds,\Lcds )$ is a $r$DS-encoder. Note that there is a unique $0$-tuple $ R_{\emptyset} \in \Cds(\emptyset) $, and by definition of $\Lcds$, $(G,S,R_{\emptyset}) \in \Lcds$ if and only if $S$ is an $r$-dominating set of $G$. Let us now prove that the equivalence relation $\sim^*_{\Eds,g,t,\mc{G}}$ is DP-friendly for $g(t)=t$.

As in Definition~\ref{def:DPfriendly}, let $G$ be a $t$-boundaried graph with boundary $A$, let $B$ be any separator $B \subseteq V(G)$ with $|B| \leq t$, and let $G_B$ be any collection of connected components of $G - B$ such that $A \cap V(G_B) \subseteq B$, which we consider as a $t$-boundaried graph with boundary $B$. We define $H$ to be the $t$-boundaried graph induced by $V(G) \setminus (V(G_B) \setminus B)$, and with boundary $B$ (that is, we forget boundary $A$) labeled as in $G_B$.
Let $G_B'$ be a $t$-boundaried graph such that $G_B \sim^*_{\Eds,g,t,\mc{G}} G_B'$. Let $G' := H \oplus G_B'$ with boundary $A$. See Fig.~\ref{fig:unglueGlue} for an illustration.

\begin{figure}[h!]
\centering 
\includegraphics[width=1.13\textwidth]{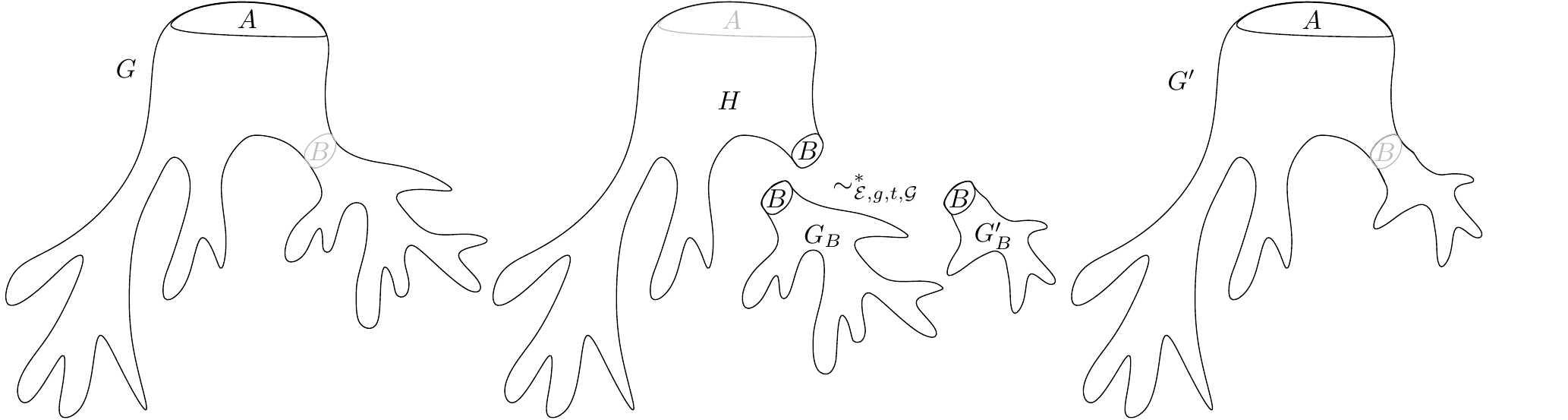}
\vspace{-.1cm}
\caption{Graphs $G$ and $G'$ in the proof of Lemma~\ref{lem:rDomCharacterizes}. \label{fig:unglueGlue}}
\end{figure}

We claim that the encoder $\Eds$ is $g$-confined for $g(t)=t$. Indeed, consider an arbitrary encoding $R_A \in \Cds(\Lambda(G))$ and the encoding $R_0$ satisfying $R_0(v)=0$ for every $v\in A$. Let $S_0\subseteq V(G)$ be a minimum-sized partial $r$-dominating set satisfying $R_0$, i.e., such that $(G,S_0,R_0)\in \Lcds$. Observe that $S_0$ also satisfies $R_A$, i.e., $(G,S_0,R_A)\in\Lcds$. It then follows that $f_{G}^{\Eds}(R_{0})  = \max_{R_A} f_{G}^{\Eds}(R_A)$.
Moreover, let $S\subseteq V(G)$ be a minimum-sized partial $r$-dominating set satisfying $R_A$, i.e., such that  $(G,S,R_A) \in \Lcds$. Then $R_0$ is also satisfied by $S \cup A$. It follows that $f_{G}^{\Eds}(R_0)  - \min_{R_A}f_{G}^{\Eds}(R_A)  \  \leqslant \ |A| \leq t $, proving that the encoder is indeed $g$-confined.

We want to show that $G \sim^*_{\Eds,g,t, \mc{G}} G'$ and that $\Delta_{\Eds,g,t}(G,G') =  \Delta_{\Eds,g,t}(G_B,G_B')$. According to Fact~\ref{fact:DPfriendly}, we can consider the relation $\sim^*_{\Eds,g,t}$ (that is, we do not need to consider the refinement with respect to the class of graphs $\mc{G}$), and due to the $g$-confinement it holds that $f_G^{\Eds,g} = f_G^{\Eds}$ for $g(t) =t$.
Hence it suffices to prove that $ f_G^{\Eds} (R_A) = f_{G'}^{\Eds} (R_A) + \Delta_{\Eds,g,t}(G_B,G_B')$ for all $R_A \in \Cds(\Lambda(G))$.

Let $R_A \in \Cds(\Lambda(G))$ be a $\Cds$-encoding defined on $A$. First assume that $f_{G}^{\Eds}(R_A)\neq + \infty$, that is, $R_A \in  \mathcal{C}_{\textsc{$r$DS},G}^{*}$. Let $S = D \cup D_H$ be a partial $r$-dominating set of size $f_{G}^{\Eds}(R_A)$ of $G$ satisfying $R_A$, with $D \subseteq V(G_B)$ and $D_H \subseteq V(H) \setminus B$.
We use $S$ to construct a $\Cds$-encoding $R_B\in \Cds(\Lambda(G_B))$ defined on $B$, satisfied by $D$ as follows. Let $v\in B$:
\begin{itemize}
\item[$\bullet$] if $v\in S$, then $R_B(v)=0$;
\item[$\bullet$] otherwise, if there is either a shortest path from $v$ to $S$ of length $i$ or a path from $v$ to any $a \in A$ such that $R_A(a) = \uparrow j$ of length $i-j$, in both cases with its first edge in $G_B$, then $R_B(v)=\downarrow i$;
\item[$\bullet$] otherwise, $R_B(v)=\uparrow i$ where $i=d_{G}(v,S)$ or $i=d_{G}(v,a) + j $ such that $R_A(a) = \uparrow j$ (the first edge of any shortest path from $v$ to $S$ is not in $G_B$).
\end{itemize}

\begin{figure}[h!]
\centering 
\includegraphics[width=1.05\textwidth]{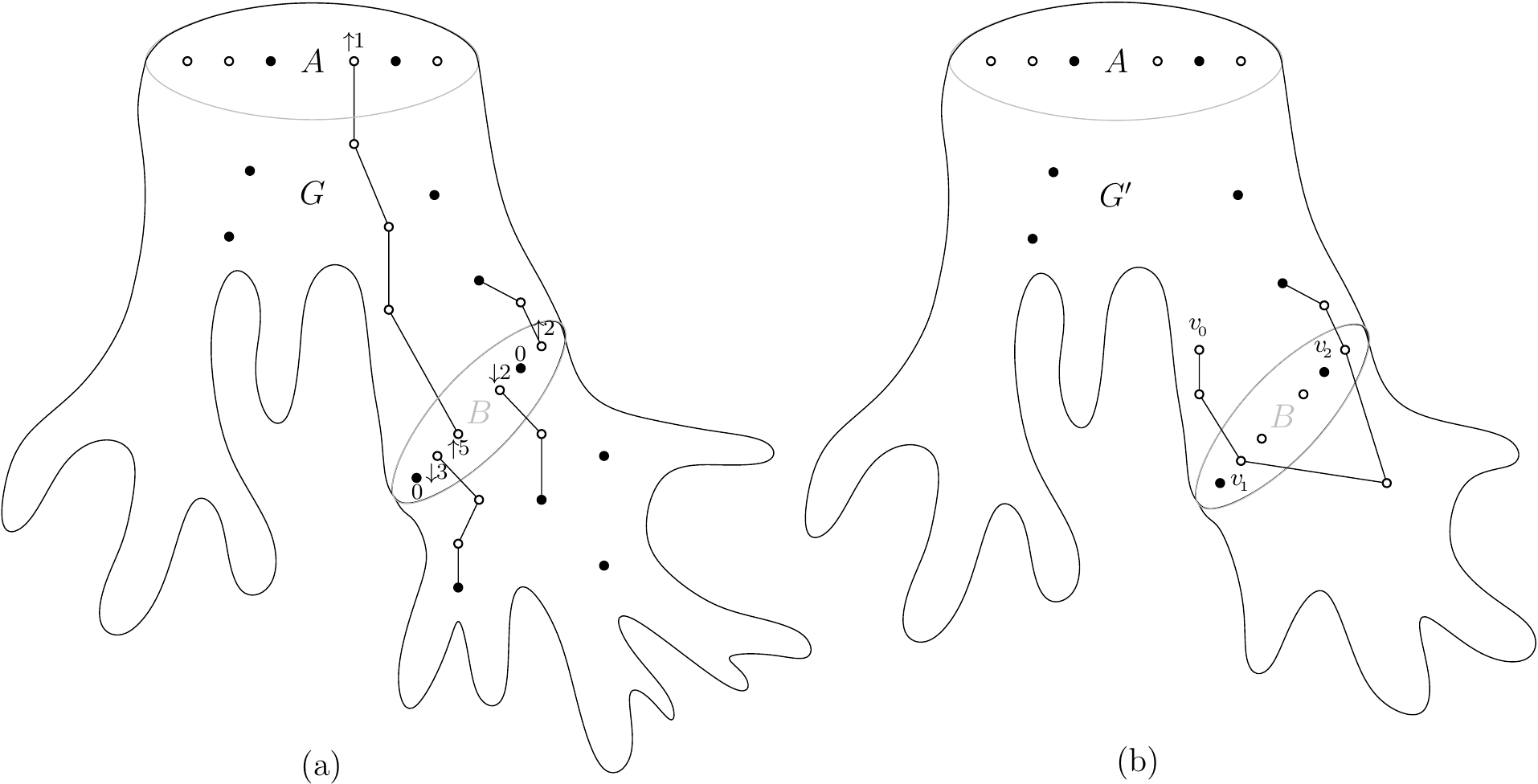}
\caption{Illustration of the proof of Lemma~\ref{lem:rDomCharacterizes}. Black vertices belong to the solution: (a) construction of the $\Cds$-encoding $R_B\in \Cds(\Lambda(G_B))$; and (b) construction of the corresponding paths.\label{fig:r-Dom}}
\end{figure}

See Fig.~\ref{fig:r-Dom}(a) for an illustration of the construction of the $\Cds$-encoding $R_B\in \Cds(\Lambda(G_B))$ described above.

Observe that by construction of $R_B$, $|D| \geqslant f_{G_B}^{\Eds}(R)$. Let $D'$ be a subset of vertices of $G_B'$ of minimum size such that $(G_B',D',R_B) \in \Lcds$, that is, $|D'| = f_{G_B'}^{\Eds}(R_B)$.  As $G_B \sim_{\Eds,g,t} G_B'$, we have $|D'|=  f_{G_B}^{\Eds}(R_B) + \Delta_{\Eds,g,t}(G_B,G_B')$ and therefore $|D' \cup D_H| =  f_{G_B}^{\Eds}(R_A) + \Delta_{\Eds,g,t}(G_B,G_B') + |D_H| \leqslant f_{G}^{\Eds}(R_A) + \Delta_{\Eds,g,t}(G_B,G_B')$.

Let us now prove that $S'=D'\cup D_H$ is a partial $r$-dominating set of $G'$ satisfying $R_A$. According to the definition of $\Eds$, we distinguish vertices in $V(G')\setminus A$ and in $A$.

We start with vertices not in $A$. For any vertex $v \in V(G')\setminus (A \cup S') $, we consider the following iterative process that builds a path of length at most $r$ from $v$ to $S'$ or a path of length at most $r-i$ from $v$ to $a \in A$ such that $R(a)= \uparrow i $. At step $j\geqslant 0$, we identify a vertex $v_j \in B$. We initially set $v_0=v$. If $v_0\in V(G_B')$, we can assume that $d_{G_B'}(v_0,D')>  r$, as otherwise we are done. As $D'$ satisfies $R_B$, this implies that $B$ contains a vertex $v_1$ such that $R_B(v_1)=\uparrow i_1$ and $d_{G_B'}(v_0,v_1)+i_1\leqslant r$. Similarly, if $v_0\in V(H)\setminus B$, we can assume that $d_H(v_0,D_H)> r$ and $d_H(v_0,a)> r-i$ for any $a \in A$ such that $R_A(a) = \uparrow i$, as otherwise we are done. As $S=D\cup D_H$ is a partial $r$-dominating set of $G$ satisfying $R_A$, any shortest path $P$ (of length at most $r$) between $v_0$ and $S$ and any path (of length $r-i$) between $v_0$ and $a \in A$ such that $R_A(a) = \uparrow i$, contains a vertex of $B$ incident to an edge of $G_B$. Let $v_1$ be the first such vertex of $P$. By definition of $R_B$, we have that $R_B(v_1)=\downarrow i_1$ with $d_H(v_0,v_1)+i_1\leqslant r$. Let us now consider $v_j$ with $j\geqslant 1$, and denote by $l_j$ the length of the path we discovered from $v_0$ to $v_j$. We need to prove that $l_j+i_j\leqslant r$ (or $l_j+i_j\leqslant r-i$ in the other case) is an invariant of the process. As we argued, it is true for $j=1$, so assume it holds at step $j$. We consider two cases:
\begin{enumerate}
\item $R_B(v_j)=\downarrow i_j$: We can assume that  $d_{G_B'}(v_j,D')>i_j$, otherwise we are done as by construction it holds that $l_j+i_j\leqslant r$ (or $l_j+i_j \leqslant r-i$ in the other case). So as $D'$ is a partial $r$-dominating set satisfying $R_B$, there exists a vertex $v_{j+1}\in B$ such that $R_B(v_{j+1})=\uparrow i_{j+1}$ and $d_{G_B'}(v_i,v_{j+1})+i_{j+1}\leqslant i_j$. As $l_{j+1}=l_j+d_{G_B'}(v_i,v_{j+1})$, it follows that $l_{j+1}+i_{j+1}\leqslant r$ (or $l_{j+1}+i_{j+1}\leqslant r-i$ in the other case). See Fig.~\ref{fig:r-Dom}(b) for an illustration of this case.
\item $R_B(v_j)=\uparrow i_j$: We can assume that $d_{H}(v_j,D_H)>i_j$ and $d_H(v_j,a)> i_j-i$ for any $a \in A$ such that $R_A(a) = \uparrow i$, otherwise we are done as by construction it holds that $l_j+i_j\leqslant r$ (or $l_j+i_j\leqslant r-i$ in the other case). As by definition of the encoding $R_B$, $d_G(v_j,S)=i_j$ , any shortest path $P$ between $v_j$ and $S$ (or $a \in A$) uses a vertex of $B$ incident to an edge of $G_B$. Let $v_{j+1}$ be the first such vertex of $P$. Then $R_B(v_{j+1})=\downarrow i_{j+1}$ with $d_H(v_j,v_{j+1})+i_{j+1}\leqslant i_j$. As $l_{j+1}=l_j+d_H(v_i,v_{j+1})$, it follows that $l_{j+1}+i_{j+1}\leqslant r$ (or $l_j+i_j\leqslant r-i$ in the other case).
\end{enumerate}
Observe that the process ends, since the parameter $r-(l_j+i_j)$ is strictly decreasing.

We now consider vertices of $A$. Note that as $\partial(G_B) = \partial(G_B')$, it holds that $\partial(G) = \partial(G')$ as well. In particular, any vertex $v \in A$ is also in $H$. If $R_A(v) = 0$ since $S= D \cup D_H$ satisfies $A$, $v \in D_H$ and hence, $v \in S'=D' \cup D_H$. If $R_A(v) = \downarrow i$, the iterative process above built a path from $v$ to $S'$ of length at most $r$, or from $v$ to $a \in A$ with $R_A(a) = \uparrow j$ of length at most $r-i-j$.

It follows that $S'=D'\cup D_H$ is a partial $r$-dominating set of size at most $f_{G}^{\Eds}(R_A)+\Delta_{\Eds,g,t}(G_B,G_B')$ satisfying $R_A$, as we wanted to prove.

Finally, assume that $f_{G}^{\Eds}(R_A) = + \infty$. Then it holds that $f_{G'}^{\Eds}(R_A) = + \infty$ as well. Indeed, suppose that $f_{G'}^{\Eds}(R_A)$ is finite. Then,  given a partial $r$-dominating set of $G'$ satisfying $R_A$,  by the argument above we could construct a partial $r$-dominating set of $G$ satisfying $R_A$, contradicting that $f_{G}^{\Eds}(R_A) = + \infty$.

Therefore, we can conclude that $G \sim^*_{\Eds,g,t} G'$, and hence the equivalence relation $\sim^*_{\Eds,g,t,\mc{G}}$ is DP-friendly for $g(t)=t$. \end{proofrDomSet}

\subsection{Construction of the kernel}
\label{sec:Construction-kernel-rDomSet}

We proceed to construct a linear kernel for \textsc{$r$-Dominating Set} when the input graph excludes a fixed apex graph $H$ as a minor. Toward this end, we use the fact that this problem satisfies the  contraction-bidimensionality and separability conditions
required in order to apply the results of Fomin \emph{et al}.~\cite{FLST10}. In the following proposition we specify  the result in~\cite[Lemma~3.3]{FLST10}
for the case of \textsc{$r$-Dominating Set} while making visible the dependance on $r$ and the size $h$ of the excluded apex graph. The polynomial-time algorithm  follows from Fomin \emph{et al}.~\cite[Lemma 3.2]{FominLRS11bidi}, whose proof makes use of the polynomial-time approximation
algorithm for the treewidth of general graphs by~Feige \emph{et al}.~\cite{FeigeHL08impr}.

%
%

\begin{proposition}\label{thm:protDecrDOMSet}
Let $r \geq 1$ be an  integer, let $H$ be an $h$-vertex apex graph, and let $\textsc{$r$DS}_H$ be the restriction of the \textsc{$r$-Dominating Set} problem to input graphs which exclude $H$ as a minor. If $(G,k) \in \textsc{$r$DS}_H$, then there exists a set $X\subseteq V(G)$ such that {$|X|=O(r \cdot f_{c}(h) \cdot k)$}
 and ${\bf tw}(G-X) = O(r \cdot (f_{c}(h))^{2})$, where $f_{c}$ is the function in Proposition~\ref{prop:tw-contraction}.
Moreover, given an instance  $(G,k)$ of $\textsc{$r$DS}_H$, there is a polynomial-time algorithm  that either finds such a set $X$ or correctly reports that $(G,k)$ is a \textsc{No}-instance.
\end{proposition}

%
%

We are now ready to present the linear kernel for \textsc{$r$-Dominating Set}.

\begin{theorem}\label{thm:KernelrDomSet}
Let $r \geq 1$ be an  integer, let $H$ be an $h$-vertex apex graph, and let $\textsc{$r$DS}_H$ be the restriction of the \textsc{$r$-Dominating Set} problem to input graphs which exclude $H$ as a minor.  Then $\textsc{$r$DS}_H$ admits a constructive linear kernel of size at most $f(r,h)\cdot k$, where $f$ is an explicit function depending only on $r$ and $h$, defined in Equation~(\ref{eq:kernelrDomSet}) below.
\end{theorem}

\begin{proof} Given an instance $(G,k)$ of $\textsc{$r$DS}_H$, we run the polynomial-time algorithm given by Proposition~\ref{thm:protDecrDOMSet} to either conclude that $(G,k)$ is a \textsc{No}-instance or to find a set $X\subseteq V(G)$
such that $|X|=O(r \cdot f_{c}(h) \cdot k)$
and ${\bf tw}(G-X) = O(r \cdot (f_{c}(h))^{2})$. In the latter case, we use the set $X$ as input to the algorithm given by Theorem~\ref{thm:protDec}, which outputs in linear time a
$(r^2 \cdot 2^{O(h \log h)}\cdot (f_{c}(h))^{3} \cdot k, O(r \cdot (f_{c}(h))^{2}))$-protrusion decomposition of $G$. We now consider the encoder $\Eds =(\Cds,L_{\mc{C}_{\textsc{$r$DS}}})$  defined in Subsection~\ref{sec:Description-encoder-rDomSet}. By Lemma~\ref{lem:rDomCharacterizes}, $\Eds$ is an $r\textsc{DS}$-encoder and $\sim^*_{\Eds,g,t,\mc{G}}$ is DP-friendly, where $\mc{G}$ is the class of $H$-minor-free graphs and $g(t)=t$. By Equation~(\ref{eq:sizeEncoderrDomSet}) in Subsection~\ref{sec:Description-encoder-rDomSet}, we have that $s_{\Eds }(t) \ \leq \ (2r+1)^t$. Therefore, we are in position to apply Corollary~\ref{corollary:withProtrusionDecomposition} and obtain a linear kernel for $\textsc{$r$DS}_H$ of size at most
\begin{equation}\label{eq:kernelrDomSet}
r^2 \cdot 2^{O(h \log h)}\cdot (f_{c}(h))^{3} \cdot b\left(\Eds, g,  O(r \cdot (f_{c}(h))^{2}) ,\mathcal{G}\right) \cdot k \ ,
\end{equation}
where $b\left(\Eds, g,  O(r \cdot (f_{c}(h))^{2}) ,\mathcal{G}\right)$ is the function defined in Lemma~\ref{lem:finiteSize}. \end{proof}

\vspace{.2cm}

It can be routinarily checked that, once the excluded apex graph $H$ is fixed, the dependance on $r$ of the multiplicative constant involved in the upper bound of~Equation~\eqref{eq:kernelrDomSet} is of the form $2^{2^{2^{O(r\cdot \log r)}}}$, that is,  it depends triple-exponentially on the integer $r$.

\section{An explicit linear kernel for \textsc{$r$-Scattered Set}}
\label{sec:rScatSet}
Let $r \geq 1$ be a fixed integer. Given a graph $G$ and a set $S \subseteq V(G)$, we say that $S$ is an \emph{$r$-independent set} if any two vertices in $S$ are at distance greater than $r$ in $G$. We define the \textsc{$r$-Scattered Set} problem, which can be seen as a generalization of  \textsc{Independent Set}, as follows.

\vspace{.4cm}
\begin{boxedminipage}{.8\textwidth}
\textsc{$r$-Scattered Set}\vspace{.1cm}

\begin{tabular}{ r l }
\textbf{~~~~Instance:} & A graph $G$  and a non-negative integer $k$. \\
\textbf{Parameter:} & The integer $k$.\\
\textbf{Question:} & Does $G$ have a $2r$-independent set of size at least $k$?\\
\end{tabular}
\end{boxedminipage}\vspace{.4cm}

Our encoder for \textsc{$r$-Scattered Set} (or equivalently, for \textsc{$r$-Independent Set}) is inspired from the proof of Fomin \emph{et al}.~\cite{BFL+09} that the problem has FII, and can be found in Subsection~\ref{sec:Description-encoder-rScatSet}. We then show how to construct the linear kernel in Subsection~\ref{sec:Construction-kernel-rScatSet}.

\subsection{Description of the encoder}
\label{sec:Description-encoder-rScatSet}

Equivalently, we proceed to present an encoder for the \textsc{$r$-Independent Set} problem, which we abbreviate as \textsc{$r$IS}. Let $G$ be a boundaried graph with boundary $\partial(G)$ and denote $I=\Lambda(G)$. The function \Csc maps $I$ to a set $\Csc(I)$ of \Csc-encodings. Each $R\in\Csc(I)$ maps $I$ to an $|I|$-tuple the coordinates of which are in one-to-one correspondence with the vertices of $\partial(G)$. The coordinate $R(v)$ of vertex $v\in \partial(G)$ is a $(|I|+1)$-tuple in $(d_S,d_{{v_1}},\dots ,d_{{v_{|I|}}})\in\{0, 1, \dots,  r,r+1\}^{|I|+1}$. For a subset $S$ of vertices of $G$, we say that $(G,S,R)$ belongs to the language $\Lcsc$ (or that $S$ is a \emph{partial $r$-independent set satisfying $R$}) if:
\begin{itemize}
\item[$\bullet$] for every pair of vertices $v\in S$ and $w\in S$, $d_G(v,w)> r$;
\item[$\bullet$] for every vertex $v\in\partial(G)$: $d_G(v,S)\geqslant d_S$ and for every $w\in\partial(G)$, $d_G(v,w)\geqslant d_w$.
\end{itemize}

As \rind is a maximization problem, by Equation~(\ref{eq:fEmax})
the function $f_G^{\Esc}$ associates to each encoding $R$ the maximum size of a partial $r$-independent set $S$ satisfying $R$. By definition of $\Esc$ it is clear that
\begin{equation}\label{eq:sizeEncoderScatSet}
s_{\Esc}(t) \leqslant (r+2) ^{ t(t+1)}.
\end{equation}

\begin{lemma}\label{lem:encoderrScatSet}
The encoder $\Esc =(\Csc,\Lcsc )$ described above is an $r\textsc{IS}$-encoder. Furthermore, if $\mc{G}$ is an arbitrary class of graphs and $g(t)=2t$, then the equivalence relation $\sim^*_{\Esc,g,t,\mc{G}}$ is DP-friendly.
\end{lemma}
\begin{proof} We first prove that $\Esc =(\Csc,\Lcsc )$ is a $r$IS-encoder. There is a unique $0$-tuple $ R_{\emptyset} \in \Csc(\emptyset) $, and by definition of $\Lcsc$, $(G,S,R_{\emptyset}) \in \Lcsc$ if and only if $S$ is an $r$-independent set of $G$.

Let $G, G'$ with boundary $A$ and $H, G_B, G_B'$ with boundary $B$ be the graphs as defined in the proof of Lemma~\ref{lem:rDomCharacterizes} (see Fig.~\ref{fig:unglueGlue}).

Let $R_0$ be the encoding satisfying $R_0(v)=(0,0,\dots, 0)$ for every $v\in B$. Observe that if $S$ is a maximum partial $r$-independent set satisfying an encoding $R_B \in \Csc(\Lambda(G_B))$, then $S$ also satisfies $R_0$. It follows that $f_{G}^{\Esc}(R_{0})  = \max_{R_B} f_{G}^{\Esc}(R_B)$ (and thus $f_{G}^{\Esc}(R_{0})=f_{G}^{\Esc,g}(R_{0})$).

We want to show that $G \sim^*_{\Esc,g,t, \mc{G}} G'$ and that $ \Delta_{\Esc,g,t}(G,G') =  \Delta_{\Esc,g,t}(G_B,G_B')$.
According to Fact~\ref{fact:DPfriendly}, it is enough to  consider the relation $\sim^*_{\Esc,g,t} $. To that aim, we will prove that $ f_G^{\Esc,g} (R_A) = f_{G'}^{\Esc,g} (R_A) + \Delta_{\Esc,g,t}(G_B,G_B')$ for all $R_A \in \Csc(\Lambda(G))$ and for $g(t) =2t$.

Let $R_A \in \Csc(\Lambda(G))$ be a $\Csc$-encoding defined on $A$. First assume that $f_{G}^{\Esc,g}(R_A)\neq - \infty$, that is, $R_A \in  \mathcal{C}_{\textsc{$r$IS},G}^{*}$. Let $S = I \cup I_H$ be a partial $r$-independent set of size $f_{G}^{\Esc,g}(R_A)$ of $G$, with $I \subseteq V(G_B)$ and $I_H \subseteq V(H) \setminus B$. An encoding $R_B \in \Csc(\Lambda(G_B))$, satisfied by $S$ is defined as follows. Let $v\in B$, then $R_B(v)=(d_S,d_{v_1},\dots, d_{v_{|B|}})$ where
\begin{itemize}
\item[$\bullet$] $d_S=d_{G_B}(v,I)$; and
\item[$\bullet$] for $i\in \{1,\ldots,|B|\}$, $d_{v_i}=\min \{d_{G_B}(v,v_i),r+1\}$ (remind that $v_i\in \partial(G)$).
\end{itemize}

\begin{fact}For the $R_B$ defined above, it holds that  $f_{G_B}^{\Esc,g}(R)\neq -\infty$, where $g(t)=2t$.
\end{fact}
\begin{proof}
Let $I_0\subseteq V(G_B)$ be a maximum partial $r$-independent set satisfying $R_0$, i.e., $(G_B,I_0,R_0) \in \Lcsc$. Let us define $I^*=I\setminus N_{r/2}(B)$, $I_0^*=I_0\setminus N_{r/2}(B)$ and $I_H^*=I_H\setminus N_{r/2}(B)$. By the pigeon-hole principle, it is easy to see that $|I_0^*| \geqslant |I_0| - t$ (otherwise $B$ would contain a vertex at distance at most $r/2$ from two distinct vertices of $I_0$).
Likewise, $|I_H^*| \geqslant |I_H| - t$. Now observe that $I_0^*\cup I_H^*$ is an $r$-independent set of $G$ and therefore $|I_0^*|+|I_H^*|\leqslant |S|$ (1) (as $S$ was chosen as a maximum $r$-independent set of $G$). As $S$ is the disjoint union of $I$ and $I_H$, we also have that $|S|\leqslant |I|+|I_H^*|+t$ (2). Combining (1) and (2), we obtain that $|I_0^*|\leqslant |I|+t$ and therefore $|I_0|\leqslant |I|+2t$. It follows that $f_{G_B}^{\Esc}(R_B)=f_{G_B}^{\Esc,g}(R_B)$, proving the fact.
\end{proof}

Observe that by construction of $R_B$, $|I| \leqslant f_{G_B}^{\Esc,g}(R_B)$. Consider a subset of vertices $I'$ of $G_B'$ of maximum size such that $(G_B',I',R_B) \in \Lcsc$, that is $|I'| = f_{G_B'}^{\Esc,g}(R_B)$.  As $G_B \eqec G_B'$, by the above claim, we have $|I'|=  f_{G_B}^{\Esc,g}(R_B) + \Delta_{\Esc,g,t}(G_B,G_B')$ and therefore $|I'\cup I_H| =  f_{G_B}^{\Esc,g}(R_B) + \Delta_{\Esc,g,t}(G_B,G_B') + |I_H| \geqslant f_{G}^{\Esc,g}(R_A) +\Delta_{\Esc,g,t}(G_B,G_B')$.

Let us prove that $S'=I'\cup I_H$ is a partial $r$-independent set of $G'$ satisfying $R_A$. Following the definition of $\Esc$, we have to verify two kinds of conditions: those on vertices in $S'$ and those on vertices in $A$. We start with vertices in $S'$. Let $P$ be a shortest path in $G$ between two vertices $v\in S'$ and $w\in S'$.
We partition $P$ into maximal subpaths $P_1,\dots ,P_q$ such that
$P_j$ (for $1\leqslant j\leqslant q$) is either a path of $G_B'$ (called a $G_B'$-path) or of $H$ (called an $H$-path). An illustration of these paths can be found in Fig.~\ref{fig:r-Ind}. If $q=1$, then $d_{G'}(v,w)> r$ follows from the fact that  $I_H$ and $I'$ are respectively $r$-independent sets of $H$ and $G_B'$ (a partial $r$-independent set is an $r$-independent set). So assume that $q>1$. Observe that every $H$-subpath is a path in $G$. By the choice of $S'$, observe that the length of every $G_B'$-subpath is at least the distance in $G_B$ between its extremities. We consider three cases:

\begin{figure}[t]
\centering 
\includegraphics[width=0.55\textwidth]{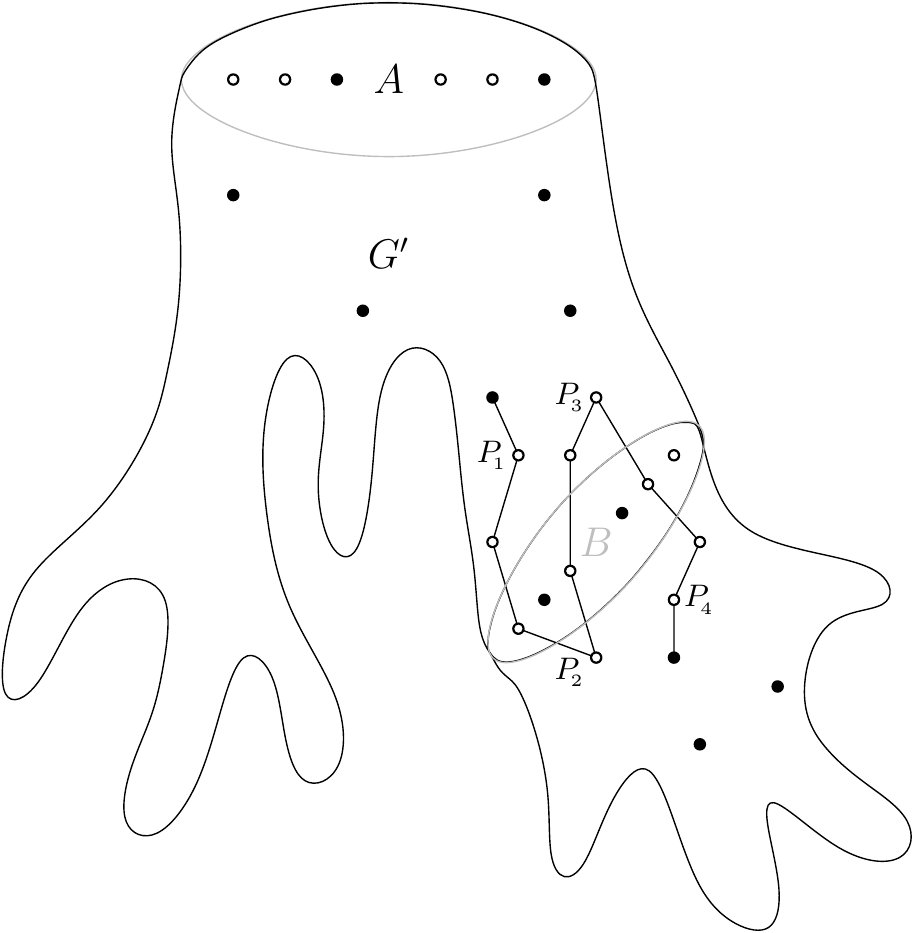}
\caption{Illustration in the proof of Lemma~\ref{lem:encoderrScatSet}. The black vertices belong to the solution. \label{fig:r-Ind}}
\end{figure}

\begin{itemize}
\item[$\bullet$] $v,w\in V(H)\setminus B$:  By the observations above, the length of $P$ is at least $d_{G}(v,w)$. As $v,w\in I_H$, we obtained that $d_{G'}(v,w)\geqslant d_{G}(v,w)>r$.
\item[$\bullet$] $v\in V(H)\setminus B$ and $w\in V(G_B')$: Let $u$ be the last vertex of $P_{q-1}$. By the same argument as in the previous case we have $d_{G'}(v,u)\geqslant d_{G}(v,u)$. Now by the choice of $S'$, observe that $d_{G_B'}(u,w)\geqslant d_{G_B}(u,I)$. So the length of $P$ is at least the distance in $G$ from $v$ to a vertex $w'\in I$, we can conclude that $d_{G'}(v,w)>r$.

\item[$\bullet$] $v,w\in V(G'_B)$: Let $u_1$ and $u_q$ be respectively the last vertex of $P_1$ and the first vertex of $P_q$. By the same argument as above, we have that $d_{G'}(u_1,u_q)\geqslant d_{G}(u_1,u_q)$. By the choice of $S'$, we have that $d_{G_B'}(u_1,v)\geqslant d_{G_B}(u_1,I)$ and $d_{G_B'}(u_q,w)\geqslant d_{G_B}(u_q,I)$. So the length of $P$ is a least the distance in $G$ between two vertices $v'\in I$ and $w'\in I$. We can therefore conclude that $d_{G'}(v,w)>r$.
\end{itemize}

We now consider vertices of $A$. Let $v \in A$ such that $R_A(v) = (d_S,d_{v_1},\dots,d_{v_{|A|}})$. Let $P$ be a shortest path in $G'$ between vertices $v\in A$ and $w\in S'$, similarly to the previous argumentation (two first items) $d_{G'}(v,w) > d_S $. Now let $P$ be a a shortest path in $G'$ between vertices $v\in A$ and $v_i\in A$ similarly to the previous argumentation (first item) $d_{G'}(v,w) > d_{v_i}$.

It follows that $S'=I'\cup I_H$ is a partial $r$-independent set of size at least $f_{G}^{\Esc,g}(R_A) + \Delta_{\Esc,g,t}(G_B,G_B')$ satisfying $R_A$, as we wanted to prove.

Finally, assume that $f_{G}^{\Esc}(R_A) = - \infty$. Then it holds that $f_{G'}^{\Esc}(R_A) = - \infty$ as well. Indeed, suppose that $f_{G'}^{\Esc}(R_A)$ is finite. Then,  given a partial $r$-independent set of $G'$ satisfying $R_A$,  by the argument above we could construct a partial $r$-independent set of $G$ satisfying $R_A$, contradicting that $f_{G}^{\Esc}(R_A) = - \infty$.

Therefore, we can conclude that $G \sim^*_{\Esc,g,t} G'$, and hence the equivalence relation $\sim^*_{\Esc,g,t,\mc{G}}$ is DP-friendly for $g(t)=2t$.\end{proof}

\subsection{Construction of the kernel}
\label{sec:Construction-kernel-rScatSet}

For constructing a linear kernel, we use the following observation, also noted in~\cite{BFL+09}. Suppose that $(G,k)$ is a \textsc{No}-instance of $r$-\textsc{Scattered Set}. Then,  if for $1 \leq i \leq k$ we greedily choose a vertex $v_i$ in $G - \bigcup_{j < i}N_{2r}[v_j]$, the graph $G - \bigcup_{1 \leq i \leq k}N_{2r}[v_i]$ is empty. Thus, $\{v_1,\ldots,v_k\}$ is a $2r$-dominating set.

\begin{lemma}[Fomin \emph{et al}.~\cite{BFL+09}]\label{lem:rScatrDom}
If $(G,k)$ is a \textsc{No}-instance of the $r$-\textsc{Scattered Set} problem,
then $(G,k)$ is a \textsc{Yes}-instance of the  $2r$-{\textsc{Dominating
Set}} problem.
\end{lemma}

We are ready to present the linear kernel for \textsc{$r$-Scattered Set} on apex-minor-free graphs.


\begin{theorem}\label{thm:KernelrScatSet}
Let $r \geq 1$ be an  integer, let $H$ be an $h$-vertex apex graph, and let $\textsc{$r$SS}_H$ be the restriction of the \textsc{$r$-Scattered Set} problem to input graphs which exclude $H$ as a minor.  Then $\textsc{$r$SS}_H$ admits a constructive linear kernel of size at most $f(r,h)\cdot k$, where $f$ is an explicit function depending only on $r$ and $h$, defined in Equation~(\ref{eq:kernelrScatSet}) below.
\end{theorem}
\begin{proof} Given an instance $(G,k)$ of $\textsc{$r$SS}_H$, we run on it the algorithm given by Proposition~\ref{thm:protDecrDOMSet} for the \textsc{$r'$-Dominating Set} problem with $r':=2r$. If the algorithm is not able to find a set $X$ of the claimed size, then by Lemma~\ref{lem:rScatrDom} we can conclude that $(G,k) \in \textsc{$r$SS}_H$. Otherwise, we use again the set $X$ as input to the algorithm given by Theorem~\ref{thm:protDec}, which outputs in linear time a
$(r^2 \cdot 2^{O(h \log h)}\cdot (f_{c}(h))^{3} \cdot k, O(r \cdot (f_{c}(h))^{2}))$-protrusion decomposition. We now consider the encoder $\mathcal{E}_{r\textsc{IS}} =(\mathcal{C}_{r\textsc{IS}},L_{\mathcal{C}_{r\textsc{IS}}})$  defined in Subsection~\ref{sec:Description-encoder-rScatSet}. By Lemma~\ref{lem:encoderrScatSet},  $\mathcal{E}_{r\textsc{IS}}$ is an $r\textsc{IS}$-encoder and $\sim^*_{\Eds,g,t,\mc{G}}$ is DP-friendly, where $\mc{G}$ is the class of $H$-minor-free graphs and $g(t)=2t$, and furthermore by Equation~(\ref{eq:sizeEncoderScatSet}) it satisfies $s_{\mathcal{E}_{r\textsc{IS}}}(t) \leq (r+2) ^{t (t+1)}$. Therefore, we are again in position to apply Corollary~\ref{corollary:withProtrusionDecomposition} and obtain a linear kernel for $\textsc{$r$SS}_H$ of size at most
\begin{equation}\label{eq:kernelrScatSet}
r^2 \cdot 2^{O(h \log h)}\cdot (f_{c}(h))^{3} \cdot b\left(\mathcal{E}_{r\textsc{IS}},g,  O(r \cdot (f_{c}(h))^{2}) ,\mathcal{G}\right) \cdot k \ ,
\end{equation}
where $b\left(\mathcal{E}_{r\textsc{IS}},g,  O(r \cdot (f_{c}(h))^{2}) ,\mathcal{G}\right)$ is the function defined in Lemma~\ref{lem:finiteSize}.
\end{proof}

\section{An explicit linear kernel for \textsc{Planar-$\mc{F}$-Deletion}}
\label{sec:PlanarFDeletion}
Let $\mc{F}$ be a finite set of graphs. We define the \textsc{$\mc{F}$-Deletion} problem as follows.

\vspace{.4cm}
\begin{boxedminipage}{.75\textwidth}
\textsc{$\mc{F}$-Deletion}\vspace{.1cm}

\begin{tabular}{ r l }
\textbf{~~~~Instance:} & A graph $G$  and a non-negative integer $k$. \\
\textbf{Parameter:} & The integer $k$.\\
\textbf{Question:} & Does $G$ have a set $S\subseteq V(G)$
such that $|S|\leqslant k$\\ &and $G-S$ is $H$-minor-free for every $H\in
\mc{F}$?\\
\end{tabular}
\end{boxedminipage}\vspace{.4cm}

When all the graphs in $\mc{F}$ are connected, the corresponding problem is called \textsc{Connected-$\mc{F}$-Deletion}, and when $\mc{F}$ contains at least one planar graph, we call it \textsc{Planar-$\mc{F}$-Deletion}. When both conditions are satisfied, the problem is called \textsc{Connected-Planar-$\mc{F}$-Deletion}. Note that \textsc{Connected-Planar-$\mc{F}$-Deletion} encompasses, in particular, \textsc{Vertex Cover} and \textsc{Feedback Vertex Set}.

Our encoder for the \textsc{$\mc{F}$-Deletion} problem uses the dynamic programming machinery developed by Adler \emph{et al}.~\cite{ADF+11}, and it is described in Subsection~\ref{sec:Description-encoder-PlanarFDeletion}. The properties of this encoder also guarantee that the equivalence relation $\sim_{\mc{G},t}$ has finite index (see the last paragraph of Subsection~\ref{sec:explicitprotrusionreplacer}). We prove that this encoder is indeed an \textsc{$\mc{F}$-Deletion}-encoder and that the corresponding equivalence relation is DP-friendly, under the constraint that all the graphs in $\mc{F}$ are {\sl connected}. Interestingly, this phenomenon concerning the connectivity seems to be in strong connection with the fact that the \textsc{$\mc{F}$-Deletion} problem has FII if all the graphs in $\mc{F}$ are connected~\cite{BFL+09,FLMS12}, but for some families $\mc{F}$ containing disconnected graphs, \textsc{$\mc{F}$-Deletion} has not FII (see~\cite{KLP+12} for an example of such family).

We then obtain a linear kernel for the problem using two different approaches. The first one, described in Subsection~\ref{sec:Description-encoder-PlanarFDeletion},  follows the same scheme as the one used in the previous sections (Sections~\ref{sec:rDomSet} and~\ref{sec:rScatSet}),
 that is, we first find a treewidth-modulator $X$ in polynomial time, and then we use this set $X$ as input to the algorithm of Theorem~\ref{thm:protDec} to find a linear protrusion decomposition of the input graph. In order to find the treewidth-modulator $X$, we need that the input graph $G$ excludes a fixed graph $H$ as a minor.

 With our second approach, which can be found in Subsection~\ref{sec:Construction-kernel-PlanarFDeletion-H-topological-minor-free}, we obtain a linear kernel on the larger class of graphs that exclude a fixed graph $H$ as a {\sl topological} minor. We provide two variants of this second approach. One possibility is to use the randomized constant-factor approximation for \textsc{Planar-$\mc{F}$-Deletion} by Fomin \emph{et al}.~\cite{FLMS12} as treewidth-modulator, which yields a randomized linear kernel that can be found in uniform polynomial time. The second possibility consists in arguing just about the {\sl existence} of a linear protrusion decomposition in \textsc{Yes}-instances, and then greedily finding large protrusions to be reduced by the protrusion replacer given by Theorem~\ref{thm:protrusionReplacement}. This yields a deterministic linear kernel that can be found in time $n^{f(H,\mc{F})}$, where $f$ is a function depending on $H$ and $\mc{F}$.

\subsection{The encoder for \textsc{$\mc{F}$-Deletion} and the index of $\sim_{\mc{G},t}$}
\label{sec:Description-encoder-PlanarFDeletion}

In this subsection we define an encoder $\EFDall$ for \textsc{$\mc{F}$-Deletion}, and along the way we will also prove that when $\mc{G}$ is the class of graphs excluding a fixed graph on $h$ vertices as a minor, then the index of the equivalence relation $\sim_{\mc{G},t}$ is bounded by $2^{t \log t} \cdot h^t \cdot 2^{h^2}$.

Recall first that a \emph{model} of a graph $F$ in a graph $G$ is a mapping $\phi$, that assigns to every edge $e \in
E(F)$ an edge $\phi(e) \in E(G)$, and to every vertex $v \in V(F)$ a non-empty
connected subgraph $\phi(v) \subseteq G$, such that
\begin{itemize}
       \item[(i)] the graphs $\{\phi(v) \mid v \in V(F)\}$ are mutually
               vertex-disjoint and the edges $\{\phi(e) \mid e \in
E(F)\}$ are pairwise distinct;
       \item[(ii)] for $e = \{u,v\} \in E(F)$, $\phi(e)$ has one
               end-vertex in $V(\phi(u))$ and the other in
$V(\phi(v))$.
\end{itemize}
Assume first for simplicity that $\mc{F}=\{F\}$ consists of a single connected graph $F$. Following~\cite{ADF+11}, we introduce a combinatorial object called \emph{rooted packing}. These objects are originally defined for branch decompositions, but we can directly translate them to tree decompositions. Loosely speaking, rooted packings capture how ``potential
models'' of $F$ intersect the
separators that the algorithm is processing. It is worth mentioning that
the notion of rooted packing is related to the notion of \emph{folio}
introduced by Robertson and Seymour in~\cite{RS95}, but more suited to dynamic programming. See~\cite{ADF+11} for more details.

Formally, let $S_F^* \subseteq V(F)$ be a subset of the vertices of the graph $F$, and let
$S_F \subseteq S_F^*$. Given a bag $B$ of a tree decomposition $(T,\mc{X})$ of the input graph $G$, we define a {\em rooted packing} of
$B$ as a quintuple $\rp =({\cal A},S_F^*,S_F,\psi,\chi)$, where ${\cal A}$ is
a (possible empty) collection of mutually disjoint non-empty subsets of
$B$ (that is, a \emph{packing} of $B$), $\psi: {\cal A} \to S_F$ is
a surjective mapping (called the \emph{rooting}) assigning vertices of $S_F$ to the sets in ${\cal A}$, and $\chi: S_F \times S_F \to \{0,1\}$ is a binary symmetric
function between pairs of vertices in $S_F$.

The intended meaning of a rooted packing $({\cal A},S_F^*,S_F,\psi,\chi)$ is as
follows. In a given separator $B$, a packing ${\cal A}$ represents the
intersection of the connected components of the potential model with
$B$. The subsets $S_F^*,S_F \subseteq V(F)$ and the function $\chi$ indicate
that we are looking in the graph $G_B$ for a potential model of $F[S_F^*]$ containing the edges between vertices in $S_F$ given by the function $\chi$. Namely, the function $\chi$ captures which edges of $F[S_F^*]$ have been
realized so far in the processed graph. Since we allow the vertex-models intersecting $B$ to be disconnected, we need to keep track of their connected components. The subset
$S_F \subseteq S_F^*$ tells us which vertex-models intersect $B$, and the
function $\psi$ associates the sets in ${\cal A}$ with the vertices in $S_F$. We
can think of $\psi$ as a coloring that colors the subsets in ${\cal A}$ with
colors given by the vertices in $S_F$. Note that several subsets in ${\cal A}$
can have the same color $u \in S_F$, which means that the vertex-model of $u$ in
$G_B$ is not connected yet, but it may get connected in further steps of the
dynamic programming. Again, see~\cite{ADF+11} for the details.

It is proved in~\cite{ADF+11} that rooted packings allow to carry out dynamic programming in order to determine whether an input graph $G$ contains a graph $F$ as a minor. It is easy to see that the number of distinct rooted packings at a bag $B$ is upper-bounded by $f(t,F):=2^{t \log t} \cdot r^t \cdot 2^{r^2}$, where $t = \tw(G)$ and $r=|V(F)|$. In particular, this proves that when $\mc{G}$ is the class of graphs excluding a fixed graph $H$ on $h$ vertices as a minor, then the index of the equivalence relation $\sim_{\mc{G},t}$ is bounded by $2^{t \log t} \cdot h^t \cdot 2^{h^2}$.

Nevertheless, in order to solve the \textsc{$\mc{F}$-Deletion} problem, we need a more complicated data structure. The intuitive reason is that it is inherently more difficult to cover {\sl all} models of a graph $F$ with at most $k$ vertices, rather than just finding one. We define $\CFD$ as the function which maps $I \subseteq \{1,\dots, t\}$ to a subspace of $\{0,1\}^{f(|I|,F)}$. That is, each $\CFD$-encoding $R \in \C(I)$ is a vector of $f(|I|,\mc{F})$ bits, which when interpreted as the tables of a dynamic programming algorithm at a given bag $B$ such that $\Lambda(G_B)=I$, prescribes which rooted packings exist in the graph $G_B$ once the corresponding vertices of the desired solution to \textsc{$\mc{F}$-Deletion} have been removed. More precisely, the language $\LFD$ contains the triples $(G,S,R)$ (recall from Definition~\ref{def:encoder} that here $G$ is a boundaried graph with $\Lambda(G) \subseteq I$, $S \subseteq V(G)$, and $R \in \mc{C}(I)$) such that the graph $G-S$ contains precisely the rooted packings prescribed by $R$ (namely, those whose corresponding bit equals 1 in $R$), and such that the graph $G-(\partial G \cup S)$ does {\sl not} contain $F$ as a minor.

When the family $\mc{F}=\{F_1,\ldots,F_{\ell}\}$ may contain more than one graph, let $f(t,\mc{F})= \sum_{i=1}^{\ell} f(t,F_i)$, and we define $\CFD$ as the function which maps $I \subseteq \{1,\dots t\}$ to a subspace of $\{0,1\}^{f(|I|,\mc{F})}$. The language is defined $\LFD$ is defined accordingly, that is, such that the graph  $G-S$ contains precisely the rooted packings of $F_i$ prescribed by $R$, for each $1 \leq i \leq \ell$, and such that the graph $G-(\partial G \cup S)$ does {\sl not} contain any of the graphs in $\mc{F}$ as a minor. By definition of $\EFD$, it clearly holds that
\begin{equation}\label{eq:sizeEncoderFDeletion}
s_{\EFD}(t)\ \leq\ 2^{f(t,F_1)} \cdot 2^{f(t,F_2)} \cdots 2^{f(t,F_{\ell})}\ =\ 2^{f(t,\mc{F})}.
\end{equation}

\noindent Assume henceforth that all graphs in the family $\mc{F}$ are {\sl connected}. This assumption is crucial because for a connected graph $F \in \mc{F}$ and a potential solution $S$, as the graph $G-(\partial G \cup S)$ does not contain $F$ as a minor, we can assume that the packing $\mc{A}$ corresponding to a potential model of $F$ rooted at $\partial G \setminus S$ is {\sl nonempty}. Indeed, as $F$ is connected, a rooted packing which does {\sl not} intersect $\partial G \setminus S$ can never be extended to a (complete) model of $F$ in $G \oplus K$ for any $t$-boundaried graph $K$.  Therefore, we can directly discard these empty rooted packings. We will use this property in the proof of Lemma~\ref{lem:encoderFdeletion} below. Note that this assumption is not safe if $F$ contains more than one connected component. As mentioned before, this phenomenon seems to be in strong connection with the fact that the \textsc{$\mc{F}$-Deletion} problem has FII if all the graphs in $\mc{F}$ are connected~\cite{BFL+09,FLMS12}, but for some families $\mc{F}$ containing disconnected graphs, \textsc{$\mc{F}$-Deletion} has not FII. 

\begin{lemma}\label{lem:encoderFdeletion}
The encoder $\EFD$ is a \textsc{Connected-$\mc{F}$-Deletion}-encoder. Furthermore, if $\mc{G}$ is an arbitrary class of graphs and $g(t)=t$, then the equivalence relation $\sim^*_{\EFD,g,t,\mc{G}}$ is DP-friendly.
\end{lemma}
\begin{proof} The fact that $\EFDall$
is a \textsc{Connected-$\mc{F}$-Deletion}-encoder follows easily from the above discussion, as if $G$ is a $0$-boundaried graph, then $\CFD(\emptyset)$ consists of a single $\CFD$-encoding $R_{\emptyset}$, and $(G,S,R_{\emptyset}) \in \LFD$ if and only if the graph $G - S$ contains none of the graphs in $\mc{F}$ as a minor.
It remains to prove that the equivalence relation $\sim^*_{\EFD,g,t,\mc{G}}$ is DP-friendly for $g(t)=t$.

The proof is similar to the proofs for \textsc{$r$-Dominating Set} and \textsc{$r$-Scattered Set}, so we will omit some details. As in the proof of Lemma~\ref{lem:rDomCharacterizes}, we start by proving that the encoder $\EFD$ for \textsc{Connected-$\mc{F}$-Deletion} is $g$-confined for the identity function $g(t)=t$. Similarly to the encoder we presented for \textsc{$r$-Dominating Set}, $\EFDall$ has the following monotonicity property. For $R_1,R_2 \in \mc{C}_{\mc{F}}(I)$ such that  $f_{G}^{\CFD}(R_1) < \infty$ and $f_{G}^{\CFD}(R_2) < \infty$,
\begin{equation}\label{Fdeletion:monotone}
 \text{if }  R_1^{-1}(0) \subseteq  R_2^{-1}(0),\ \text{ then }\ f_{G}^{\CFD}(R_1) \leq f_{G}^{\CFD}(R_2),
\end{equation}
where for $i \in \{ 1,2\}$, $R_i^{-1}(0)$ denotes the set of  rooted packings whose corresponding bit equals $0$ in $R_i$. Indeed, Equation~(\ref{Fdeletion:monotone}) holds because any solution $S$ in $G$ that covers all the rooted packings forbidden by $R_2$ also covers those forbidden by $R_1$ (as by hypothesis $R_1^{-1}(0) \subseteq  R_2^{-1}(0)$), so it holds that $f_{G}^{\CFD}(R_1) \leq f_{G}^{\CFD}(R_2)$. Let $R_0=\{0,0,\ldots,0\}$ be the $\CFD(I)$-encoding will all the bits set to $0$. The key observation is that, since  each graph in $\mc{F}$ is \emph{connected}, by the discussion above the lemma we can assume that each packing $\mc{A}$ in a rooted packing is nonempty. This implies that if $R \in \CFD(I)$ such that  $(G,S,R) \in \LFD$ for some set $S \subseteq V(G)$, then $(G,S \cup \partial G,R_0) \in L_{\mc{C}_{\mc{F}}}$. In other words, any solution $S$ for an arbitrary $\CFD$-encoder $R$ can be transformed into a solution for $R_0$ by adding a set of vertices of size at most $|\partial(G)| \leq t$. As by Equation~(\ref{Fdeletion:monotone}), for any $\CFD$-encoding $R$ with $f_{G}^{\CFD}(R) < \infty$, it holds that $f_{G}^{\CFD}(R) \leq f_{G}^{\CFD}(R_0)$, it follows that for any graph $G$ with $\Lambda(G)=I$, 
$$
\max_{R \in \mathcal{C}_{\mc{F}\textsc{D},G}^{*}(I)}f_{G}^{\EFD}(R)\  -\   \min_{R \in \mathcal{C}_{\mc{F}\textsc{D},G}^{*}(I)}f_{G}^{\EFD}(R)  \  \leq \ t,\  \  \ \text{as we wanted to prove.}$$

Once we have that $\EFDall$ is $g$-confined, the proof goes along the same lines of that of Lemma~\ref{lem:rDomCharacterizes}. That is, the objective is to show that, in the setting depicted in Fig.~\ref{fig:unglueGlue}, $G \sim^*_{\EFD,g,t, \mc{G}} G'$ (due to Fact~\ref{fact:DPfriendly}) and  $\Delta_{\EFD,g,t}(G,G') =  \Delta_{\EFD,g,t}(G_B,G_B')$. Due to the $g$-confinement, it suffices to prove that $ f_G^{\EFD} (R_A) = f_{G'}^{\EFD} (R_A) + \Delta_{\EFD,g,t}(G_B,G_B')$ for all $R_A \in \CFD(\Lambda(G))$. Since $G_B \sim^*_{\EFD,g,t} G_B'$, the definition of $\EFD$ it implies that the graphs $G_B$ and $G_B'$ contain exactly the same set of rooted packings, so their behavior with respect to $H$ (see Fig.~\ref{fig:unglueGlue}) in terms of the existence of models of graphs in $\mc{F}$ is exactly the same. For more details, it is proved in~\cite{ADF+11} that using the  encoder $\EFDall$, the tables of a given bag in a tree- or branch-decomposition can indeed be computed from the tables of their children. Therefore, we have that $G \sim^*_{\EFD,g,t, \mc{G}} G'$.  Finally, the fact that $f_G^{\EFD} (R_A) = f_{G'}^{\EFD} (R_A) + \Delta_{\EFD,g,t}(G_B,G_B')$ can be easily proved by noting that any set $S \in V(G)$ satisfying $R_A$ can be transformed into a set $S' \in V(G')$ satisfying $R_A$ such that $|S'| \leq |S| - \Delta_{\EFD,g,t}(G_B,G_B')$ (by just replacing $S \cap V(G_B)$ with the corresponding set of vertices in $V(G_B')$, using that $G_B \sim^*_{\EFD,g,t} G_B'$), and vice versa.
\end{proof}


\subsection{Construction of the kernel on $H$-minor-free graphs}
\label{sec:Construction-kernel-PlanarFDeletion-H-minor-free}

The objective of this subsection is to prove the following theorem.

\begin{theorem}\label{thm:KernelFDeletion}
Let $\mc{F}$ be a finite set of connected graphs containing at least one $r$-vertex planar graph $F$, let $H$ be an $h$-vertex graph, and let $\textsc{CP}\mc{F}\textsc{D}_H$ be the restriction of the \textsc{Connected-Planar-$\mc{F}$-Deletion} problem to input graphs which exclude $H$ as a minor.  Then $\textsc{CP}\mc{F}\textsc{D}_H$ admits a constructive linear kernel of size at most  $f(r,h)\cdot k$, where $f$ is an explicit function depending only on $r$ and $h$, defined in Equation~(\ref{eq:kernelFDeletion}) below.
\end{theorem}

Similarly to the strategy that we presented in Subsection~\ref{sec:Construction-kernel-rDomSet} for \textsc{$r$-Dominating Set}, in order to construct a linear kernel for \textsc{Connected-Planar-$\mc{F}$-Deletion} when the input graph excludes a fixed graph $H$ as a minor, we use the fact that this problem satisfies the  minor-bidimensionality and separability conditions
required in order to apply the results of Fomin \emph{et al}.~\cite{FLST10}. Namely, in the following proposition we specify  the result in~\cite[Lemma~3.3]{FLST10}
for the case of \textsc{Planar-$\mc{F}$-Deletion} while making visible the dependance on $r$ and the size $h$ of the excluded graph. Again, the polynomial-time algorithm  follows from Fomin \emph{et al}.~\cite[Lemma 3.2]{FominLRS11bidi}.

\begin{proposition}\label{thm:protDecFDeletion}
Let $\mc{F}$ be a finite set of graphs containing at least one $r$-vertex planar graph $F$, let $H$ be an $h$-vertex graph, and let $\textsc{P}\mc{F}\textsc{D}_H$ be the restriction of the \textsc{Planar-$\mc{F}$-Deletion} problem to input graphs which exclude $H$ as a minor. If $(G,k) \in \textsc{P}\mc{F}\textsc{D}_H$, then there exists a set $X\subseteq V(G)$
 such that {$|X|=O(r \cdot f_{m}(h) \cdot k)$}
 and ${\bf tw}(G-X) = O(r \cdot (f_{m}(h))^{2})$, where $f_{m}$ is the function in Proposition~\ref{prop:tw-minor}.
Moreover,  given an instance  $(G,k)$ of $\textsc{P}\mc{F}\textsc{D}_H$, there is a polynomial-time algorithm  that either finds such a set $X$ or correctly reports that $(G,k)$ is a \textsc{No}-instance.

%
%
%
%
\end{proposition}

We are ready to present a linear kernel for \textsc{Connected-Planar-$\mc{F}$-Deletion} when the input graph excludes a fixed graph $H$ as a minor.\vspace{.35cm}


\begin{proofDeletion}
The proof is very similar to the one of Theorem~\ref{thm:KernelrDomSet}. Given an instance $(G,k)$, we run the polynomial-time algorithm given by Proposition~\ref{thm:protDecFDeletion} to either conclude that $(G,k)$ is a \textsc{No}-instance or to find a set $X\subseteq V(G)$
such that $|X|=O(r \cdot f_{m}(h) \cdot k)$
and ${\bf tw}(G-X) = O(r \cdot (f_{m}(h))^{2})$. In the latter case, we use the set $X$ as input to the algorithm given by Theorem~\ref{thm:protDec}, which outputs in linear time a
$(r^2 \cdot 2^{O(h \log h)}\cdot (f_{m}(h))^{3} \cdot k, O(r \cdot (f_{m}(h))^{2}))$-protrusion decomposition of $G$. We now consider the encoder  $\EFDall$  defined in Subsection~\ref{sec:Description-encoder-PlanarFDeletion}. By Lemma~\ref{lem:encoderFdeletion}, $\EFD$ is a $\textsc{CP}\mc{F}\textsc{D}_H$-encoder and $\sim^*_{\EFD,g,t,\mc{G}}$ is DP-friendly, where $g(t)=t$ and $\mathcal{G}$ is the class of $H$-minor-free graphs. An upper bound on $s_{\EFD}(t)$ is given in Equation~(\ref{eq:sizeEncoderFDeletion}). Therefore, we are in position to apply Corollary~\ref{corollary:withProtrusionDecomposition} and obtain a linear kernel for $\textsc{CP}\mc{F}\textsc{D}_H$ of size at most
\begin{equation}\label{eq:kernelFDeletion}
r^2 \cdot 2^{O(h \log h)}\cdot (f_{m}(h))^{3} \cdot b\left(\EFD,g, O(r \cdot (f_{m}(h))^{2}) ,\mathcal{G}\right) \cdot k \ ,
\end{equation}
where $b\left(\EFD,g, O(r \cdot (f_{m}(h))^{2}) ,\mathcal{G}\right)$ is the function defined in Lemma~\ref{lem:finiteSize}. 
\end{proofDeletion}


\subsection{Linear kernels on $H$-topological-minor-free graphs}
\label{sec:Construction-kernel-PlanarFDeletion-H-topological-minor-free}

In this subsection we explain how to obtain linear kernels for \textsc{Planar-$\mc{F}$-Deletion} on graphs excluding a topological minor. We first describe a uniform randomized kernel and then a nonuniform deterministic one. We would like to note that in the case that $\mc{G}$ is the class of graphs excluding a fixed $h$-vertex graph $H$ as a topological minor, by using a slight variation of the rooted packings described in Subsection~\ref{sec:Description-encoder-PlanarFDeletion} it can be proved, using standard dynamic techniques, that the index of the equivalence relation $\sim_{\mc{G},t}$ is also upper-bounded by $2^{t \log t} \cdot h^t \cdot 2^{h^2}$.

Before presenting the uniform randomized kernel, we need the following two results.

\begin{theorem}[Fomin \emph{et al}.~\cite{FLMS12}]\label{thm:approx}
The optimization version of the \textsc{Planar-$\mc{F}$-Deletion} problem admits a randomized constant-factor approximation.
\end{theorem}

\begin{theorem}[Leaf and Seymour~\cite{LeSe12}]\label{thm:boundtw}
For every simple planar graph $F$ on $r$ vertices, every $F$-minor-free graph $G$ satisfies $\tw(G) \leq 2^{15r + 8r \log r}$.
\end{theorem}

\begin{theorem}\label{thm:KernelFDeletionRandomized}
Let $\mc{F}$ be a finite set of connected graphs containing at least one $r$-vertex planar graph $F$, let $H$ be an $h$-vertex graph, and let $\textsc{CP}\mc{F}\textsc{D}_{H\mbox{\footnotesize{\emph{-top}}}}$ be the restriction of the \textsc{Connected-Planar-$\mc{F}$-Deletion} problem to input graphs which exclude $H$ as a topological minor.  Then $\textsc{CP}\mc{F}\textsc{D}_{H\mbox{\footnotesize{\emph{-top}}}}$ admits a linear randomized kernel of size at most  $f(r,h)\cdot k$, where $f$ is an explicit function depending only on $r$ and $h$, defined in Equation~(\ref{eq:kernelFDeletionRandom}) below.
\end{theorem}
\begin{proof}
Given an instance $(G,k)$ of $\textsc{CP}\mc{F}\textsc{D}_{H\mbox{\footnotesize{\emph{-top}}}}$, we first run the randomized polynomial-time approximation algorithm given by Theorem~\ref{thm:approx}, which achieves an expected constant ratio $c_{\mc{F}}$. If we obtain a solution $X \subseteq V(G)$ such that $|X| > c_{\mc{F}} \cdot k$, we declare that $(G,k)$ is a \textsc{No}-instance.  Otherwise, if $|X| \leq c_{\mc{F}} \cdot k$, we use the set $X$ as input to the algorithm given by Theorem~\ref{thm:protDec}. As by Theorem~\ref{thm:boundtw}  we have that $\tw(G-X) \leq 2^{15r + 8r \log r}$, we obtain in this way a $\left( c_{\mc{F}} \cdot 40h^2 \cdot 2^{15r + 8r\log r+5h\log h}\cdot k,   2^{15r + 8r\log r +1} + h \right)$-protrusion decomposition of $G$.
We now consider again the encoder  $\EFDall$  defined in Subsection~\ref{sec:Description-encoder-PlanarFDeletion}, and by Corollary~\ref{corollary:withProtrusionDecomposition} we obtain a kernel of size at most
\begin{equation}\label{eq:kernelFDeletionRandom}
\left(1 +  b\left(\EFD, g,2^{15r + 8r\log r +1} + h,\mathcal{G}\right)\right) \cdot  \left(c_{\mc{F}} \cdot 40h^2 \cdot 2^{15r + 8r\log r+5h\log h}\right)\cdot k \ ,
\end{equation}
where $b\left(\EFD,g, 2^{15r + 8r\log r +1} + h ,\mathcal{G}\right)$ is the function defined in Lemma~\ref{lem:finiteSize} and $\mathcal{G}$ is the class of $H$-topological-minor-free graphs.\end{proof}

We finally present a deterministic kernel, whose drawback is that the running time is nonuniform on $\mc{F}$ and $H$.

\begin{theorem}\label{thm:KernelFDeletionNonUniform}
Let $\mc{F}$ be a finite set of connected graphs containing at least one $r$-vertex planar graph $F$, let $H$ be an $h$-vertex graph, and let $\textsc{CP}\mc{F}\textsc{D}_{H\mbox{\footnotesize{\emph{-top}}}}$ be the restriction of the \textsc{Connected-Planar-$\mc{F}$-Deletion} problem to input graphs which exclude $H$ as a topological minor.  Then $\textsc{CP}\mc{F}\textsc{D}_{H\mbox{\footnotesize{\emph{-top}}}}$ admits a linear kernel of size at most  $f(r,h)\cdot k$, where $f$ is an explicit function depending only on $r$ and $h$, defined in Equation~(\ref{eq:kernelFDeletionNonUniform}) below.
\end{theorem}

\begin{proof} The main observation is that if $(G,k) \in \textsc{CP}\mc{F}\textsc{D}_{H\mbox{\footnotesize{-top}}}$, then there exists a set $X \subseteq V(G)$ with $|X| \leq k$ such that $G-X$ is ${\mc{F}}$-minor-free. In particular, by Theorem~\ref{thm:boundtw} it holds that
$\tw(G-X) \leq 2^{15r + 8r \log r}$. Therefore, we know by Theorem~\ref{thm:protDec} that if $(G,k) \in \textsc{CP}\mc{F}\textsc{D}_{H\mbox{\footnotesize{-top}}}$, then $G$ admits a $(40 \cdot h^2 \cdot 2^{15r + 8r \log r +5h \log h}\cdot k, 2^{15r + 8r \log r+1}+ h)$-protrusion decomposition. Nevertheless, we do not have tools to efficiently find such linear decomposition. However, we use that, as observed in~\cite{BFL+09},
a $t$-protrusion of size more than a prescribed number $x$  in an $n$-vertex graph can be found in $n^{O(t)}$ steps, it if exists. Our kernelization algorithm proceeds as follows. We try to find a $(2^{15r + 8r \log r+1}+ h)$-protrusion $Y$ of size strictly larger than $x:=b(\EFD,g,2^{15r + 8r \log r+1}+ h,\mathcal{G})$, where $\EFD$ is the encoder for \textsc{$\mc{F}$-Deletion} described in Subsection~\ref{sec:Description-encoder-PlanarFDeletion}, $b\left(\EFD, g,2^{15r + 8r\log r +1} + h ,\mathcal{G}\right)$ is the function defined in Lemma~\ref{lem:finiteSize}, and $\mathcal{G}$ is the class of $H$-topological-minor-free graphs. If we succeed, we apply the protrusion replacement algorithm given by Theorem~\ref{thm:protrusionReplacement} and replace $Y$ with another $t$-boundaried graph $Y'$ such that $|Y'| \leq b\left(\EFD, g,2^{15r + 8r\log r +1} + h ,\mathcal{G}\right)$.
The algorithm continues as far as we are able to find such large protrusion. At the end of this procedure, we either obtain an equivalent instance of size at most
\begin{equation}\label{eq:kernelFDeletionNonUniform}
b\left(\EFD, g,2^{15r + 8r\log r +1} + h ,\mathcal{G}\right) \cdot  40 \cdot h^2 \cdot 2^{15r + 8r \log r +5h \log h}\cdot k\ ,
\end{equation}
or otherwise we can correctly declare that $(G,k)$ is a \textsc{No}-instance. This kernelization algorithm runs in time  $n^{O(2^{15r + 8r \log r+1}+ h)}$.\end{proof}

To conclude this section, we would like to note that the recent results of  Chekuri and Chuzhoy\ci{ChCh13} show that in Theorem~\ref{thm:boundtw},  the inequality $\tw(G) \leq 2^{15r + 8r \log r}$ can be replaced with $\tw(G) = r^{O(1)}$. This directly implies that in Equations~(\ref{eq:kernelFDeletionRandom}) and~(\ref{eq:kernelFDeletionNonUniform}), as well as in the running time of the algorithm of Theorem~\ref{thm:KernelFDeletionNonUniform}, the term $2^{15r + 8r \log r}$ can be replaced with $r^{O(1)}$. Nevertheless, we decided to keep the current bounds in order to be able to give explicit constants.

\section{Further research}
\label{sec:further}
The methodology for performing explicit protrusion replacement via dynamic programming that we have presented is quite general, and it could also be used to obtain polynomial kernels (not necessarily linear).
We have restricted ourselves to vertex-certifiable problems, but is seems plausible that our approach could be also extended to edge-certifiable problems or  to problems on directed graphs.


We have presented in Section~\ref{sec:PlanarFDeletion} a linear kernel for \textsc{Connected-Planar-$\mc{F}$-Deletion} when the input graph excludes a fixed graph $H$ as a (topological) minor.  The \textsc{Planar-$\mc{F}$-Deletion} problem is known to admit a polynomial kernel on general graphs~\cite{FLMS12}. Nevertheless, this kernel has size $O(k^c)$, where $c$ is a constant depending on $\mc{F}$ that is upper-bounded by $2^{2^{r^{10}}}$, where $r$ is the size of a largest graph in $\mc{F}$. The existence of a \emph{uniform} polynomial kernel (that is, a polynomial kernel whose degree does not depend on the family $\mc{F}$) for \textsc{Planar-$\mc{F}$-Deletion} on general graphs remains open.

As mentioned above, our linear kernel for \textsc{Planar-$\mc{F}$-Deletion}
requires that all graphs in the family $\mc{F}$ are {\sl connected}. It would be interesting to get rid of this assumption. On the other hand, in the linear kernel for \textsc{Connected-Planar-$\mc{F}$-Deletion} on $H$-topological-minor-free graphs given in Theorem~\ref{thm:KernelFDeletionRandomized}, the randomization appears because we use the randomized constant-factor approximation for \textsc{Planar-$\mc{F}$-Deletion} on general graphs~\cite{FLMS12}, but for our kernel to be deterministic, it would be enough with a constant-factor approximation on $H$-topological-minor-free graphs, which is not known.

All the applications examined in this paper concerned parameterized problems tuned
by a secondary parameter, i.e., $r$ for the case of \textsc{$r$-Dominating Set} and \textsc{$r$-Scattered Set}
and the size of the graphs in ${\cal F}$ for the case of ${\cal F}$-\textsc{Deletion}.
In all kernels derived for these problems,
the dependance on this secondary parameter is triple-exponential, while
the dependance on the excluded graph $H$ involves the functions $f_m$ and $f_c$ defined in Section~\ref{sec:prelim}. Two questions arise:
\begin{itemize}
\item[$\bullet$] Extend our results to larger graph classes and more general problems. Also, improve the dependance of the size of the kernels on the ``meta-parameters'' associated with the problems (that is, $r$, $\mc{F}$, and $H$). Probably the recent results of  Chekuri and Chuzhoy\ci{ChCh13} can be used in this direction. Moreover, provide refinements of this framework that can lead to reasonable explicit bounds for the kernels for particular problems.
\item[$\bullet$] Examine to what extent this exponential dependance is unavoidable under some
assumptions based on automata theory or (parameterized) complexity theory.
We suspect that the unification
between dynamic programming and kernelization that we propose in this paper might
offer a common understanding of the lower bounds
in the running time of dynamic programming algorithms for certain problems (see~\cite{LokshtanovMS11slig,LokshtanovMS11know})
and the sizes of their corresponding kernels (see for instance~\cite{HolgerM10sati,BodlaenderJK11cros,Bodlaender09kern,BodlaenderDFH08onpr}). Finally, we refer the reader to~\cite{ACPS93} for constructibility issues of algebraic graph reduction.


%
\end{itemize}

\vspace{.15cm}

{\small \noindent \textbf{Acknowledgement.} We wish to thank Fedor V. Fomin and Saket Saurabh for their advises and comments  on this work.  Also, we are particularly indebted to Daniel Lokshtanov for his insightful remarks and suggestions. 
Finally, we would  like to thank the anonymous referees of the conference version of this paper for helpful remarks that improved the presentation of the manuscript. } 

\bibliographystyle{abbrv}


\end{document}